\documentclass[a4paper,11pt]{article}

\usepackage[margin=1.125in]{geometry}
\usepackage{xcolor}
\definecolor{DarkGreen}{rgb}{0.1,0.5,0.1}
\definecolor{DarkRed}{rgb}{0.5,0.1,0.1}
\definecolor{DarkBlue}{rgb}{0.1,0.1,0.5}

\usepackage[small]{caption}
\usepackage[pdftex]{hyperref}
\hypersetup{
    unicode=false,          
    pdftoolbar=true,        
    pdfmenubar=true,        
    pdffitwindow=false,      
    pdfnewwindow=true,      
    colorlinks=true,       
    linkcolor=DarkBlue,          
    citecolor=DarkGreen,        
    filecolor=DarkGreen,      
    urlcolor=DarkBlue,          
    pdftitle={},
    pdfauthor={},    
}
\usepackage{amssymb,amsmath,amsthm,amsfonts,enumitem}
\usepackage{fullpage,nicefrac,comment}
\usepackage{tikz}
\usetikzlibrary{arrows,decorations.pathmorphing,decorations.shapes,snakes,patterns,positioning}
\usepackage[ruled,linesnumbered]{algorithm2e}

\renewcommand {\a}{\alpha}
\renewcommand {\b}{\beta}

\renewcommand {\l}{\ell}

\newcommand   {\Lst}{\mathcal{L}_{s, t}}
\renewcommand {\L}{\mathcal{L}}

\newcommand{\cC}{\ensuremath{\mathcal{C}}}

\newcommand{\F}{{\mathbb F}}

\newcommand{\inabs}[1]{\left|#1\right|}

\newcommand{\inset}[1]{\left\{#1\right\}}

\newcommand{\inparen}[1]{\left(#1\right)}

\newcommand{\suchthat}{\,:\,}

\newcommand{\eval}{\mathsf{eval}}

\newcommand{\divides}{\mbox{ {\Large $|$} } }

\newcommand{\eps}{\varepsilon}
\renewcommand{\epsilon}{\varepsilon}

\newtheorem{theorem}{Theorem} 
\newtheorem{lemma}[theorem]{Lemma} 
\newtheorem{definition}{Definition}

\newtheorem{observation}[theorem]{Observation}

\newtheorem{corollary}[theorem]{Corollary} 

\newtheorem{remark}{Remark}
\newtheorem{claim}[theorem]{Claim}

\newtheorem{question}{Question}

\newtheorem{construction}{Construction}

\title{Locality via Partially Lifted Codes} 
\author{S. Luna Frank-Fischer\thanks{Computer Science Department, Stanford University.  MW's research is supported in part by NSF grants DMS-1400558 and CCF-1657049.} \and Venkatesan Guruswami\thanks{Computer Science Department, Carnegie Mellon University.  Research supported in part by NSF grants CCF-1563742 and CCF-1422045.} \and Mary Wootters\footnotemark[1]}

\begin{document}
\maketitle

\begin{abstract}
In error-correcting codes, \em locality \em refers to several different ways of quantifying how easily a small amount of information can be recovered from encoded data.  In this work, we study a notion of locality called the $s$-Disjoint-Repair-Group Property ($s$-DRGP).  This notion can interpolate between two very different settings in coding theory: that of Locally Correctable Codes (LCCs) when $s$ is large---a very strong guarantee---and Locally Recoverable Codes (LRCs) when $s$ is small---a relatively weaker guarantee.  This motivates the study of the $s$-DRGP for intermediate $s$, which is the focus of our paper.  We construct codes in this parameter regime which have a higher rate than previously known codes.  Our construction is based on a novel variant of the \em lifted codes \em of Guo, Kopparty and Sudan.  Beyond the results on the $s$-DRGP, we hope that our construction is of independent interest, and will find uses elsewhere.
\end{abstract}

\newpage
\setcounter{page}{1}

\section{Introduction}
In the theory of error correcting codes, \em locality \em refers to several different ways of quantifying how easily a small amount of information can be recovered from encoded data.   Slightly more formally, suppose that $\cC \subset \Sigma^N$ is a \em code \em over an alphabet $\Sigma$; that is, $\cC$ is any subset of $\Sigma^N$.  Suppose that $c \in \cC$, and that we have query access to a noisy version $\tilde{c}$ of $c$.  We are tasked with finding $c_i \in \Sigma$ for some $i \in [N]$. 
Informally, we say that the code $\cC$ exhibits good \em locality \em if we may recover $c_i$ using very few queries to $\tilde{c}$.  Of course, the formal definition of locality in this set-up depends on the nature of the noise, and the question is interesting for a wide variety of noise models.

One (extremely strong) model of noise is that handled by \em Locally Correctable Codes \em (LCCs), which have been extensively studied in theoretical computer science for over 15 years.  This model is motivated by a variety of applications in theoretical computer science and cryptography, including probabilistically checkable proofs (PCPs), derandomization, and private information retrieval (PIR); we refer the reader to~\cite{Y10} for an excellent survey on LCCs.  In the LCC setting, $\tilde{c} \in \Sigma^N$ has a constant fraction of errors: that is, we are guaranteed that the Hamming distance between $\tilde{c}$ and $c$ is no more than $\delta N$, for some small constant $\delta > 0$.  The goal is to recover $c_i$ with high probability from $Q = o(N)$ randomized queries to $\tilde{c}$. 

Another (much weaker) model of noise is that handled by \em Locally Recoverable Codes \em (LRCs) and related notions, which have been increasingly studied recently motivated by applications in distributed storage~\cite{HCL13,GHSY12,xoring}.  In this model, $\tilde{c} \in (\Sigma \cup \{\bot\})^N$ has a constant \em number \em of \em erasures: \em  that is, we are guaranteed that the number of $\bot$ symbols in $\tilde{c}$ is at most some constant $e = O(1)$, and further that $c_i = \tilde{c}_i$ whenever $\tilde{c}_i \neq \bot$.  As before, the goal is to recover $c_i$ using as few queries as possible to $\tilde{c}$.  Batch codes~\cite{IKO04,DGRS14} and PIR codes~\cite{codedpir, BE16} are other variants that are interesting in this parameter regime.

A key question in both of these lines of work is how to achieve these recovery guarantees with as high a \em rate \em as possible.  The rate of a code $\cC \in \Sigma^N$ is defined to be the ratio $\log_{|\Sigma|}(|\cC|) / N$; it captures how much information can be transmitted using such a code.  In other words, given $N$, we seek to find a $\cC \subseteq \Sigma^N$ with good locality properties, so that $|\cC|$ is as large as possible.

In the context of the second line of work above, recent work~\cite{WZ14,RPDV14,TB_mult_rec,TBF16,AY17,codedpir} has studied (both implicitly and explicitly)  the trade-off between rate and something
called the $s$-Disjoint-Repair-Group-Property ($s$-DRGP) for small $s$.   Informally, $\cC$ has the $s$-DRGP if any symbol $c_i$ can be obtained from $s$ disjoint query sets $c|_{S_1}, c|_{S_2}, \ldots, c|_{S_s}$ for $S_i \subseteq [N]$.  (Notice that there is no explicit bound on the size of these query sets, just that they must be disjoint).

One observation which we will make below is that the $s$-DRGP provides a natural way to interpolate between the first (LCC) setting and the second (LRC) setting above.  More precisely, while the LRC setting corresponds to small $s$ (usually, $s = O(1)$), the LCC setting is in fact equivalent to the case when $s = \Omega(N)$.
This observation motivates the study of \em intermediate \em $s$, which is the goal in this paper. 

\paragraph{Contributions.}
Before we give a more detailed overview of previous work, we outline the main contributions of this paper.
\begin{enumerate}
	\item \textbf{Constructions of codes with the $s$-DRGP for intermediate $s$.} We give a construction of a family of codes which have the $s$-DRGP for $s \sim N^{1/4}$.  Our construction can achieve a higher rate than previous constructions with the same property.
	\item \textbf{A general framework, based on partially lifted codes.} Our codes are based on a novel variant of the \em lifted codes \em of Guo, Kopparty and Sudan~\cite{GKS13}.  In that work, with the goal of obtaining LCCs, the authors showed how to construct affine-invariant codes by a ``lifting" operation.  In a bit more detail, their codes are multivariate polynomial codes, whose entries are indexed by $\F_q^m$ (so $N = q^m$).  These codes have the property that, the restriction of each codeword to every line in $\F_q^m$ is a codeword of a suitable univariate polynomial code.  (For example, a Reed-Muller code is a \em subset \em of a lift of a Reed-Solomon code; the beautiful insight of~\cite{GKS13} is that in fact the lifted code may be much larger.)

In our work, we introduce a version of the lifting operation where we only require that the restriction to \em some \em lines lie in the smaller code, rather than the restriction to \em all \em lines; we call such codes ``partially lifted codes."  This partial lifting operation potentially allows for higher-rate codes, and, as we will see, it naturally gives rise to codes with the $s$-DRGP.

One of our main contributions is the introduction of these codes, as well as some machinery which allows us to control their rate.  We instantiate this machinery with a particular example, in order to obtain the construction advertised above.  We can also recover previous results in the context of this machinery.

	\item\textbf{Putting the study of the $s$-DRGP in the context of LRCs and LCCs.}  While the $s$-DRGP has been studied before, to the best of our knowledge, it is not widely viewed as a way to interpolate between the two settings described above.  One of the goals of this paper is to highlight this property and its potential importance to our understanding of locality, both from the LRC/batch code/PIR code side of things, and from the LCC side. 
\end{enumerate}

\subsection{Background and related work}

As mentioned above, in this work we study the $s$-Disjoint-Repair-Group Property ($s$-DRGP).
We begin our discussion of the $s$-DRGP with some motivation from the LRC end of the spectrum, from applications in distributed storage.  The following model is common in distributed storage: imagine that each server or node in a distributed storage system is holding a single symbol of a codeword $c \in \cC$.  Over time, nodes fail, usually one at a time, and we wish to repair them (formally, recovering $c_i$ for some $i$).  Moreover, when they fail, it is clear that they have failed.  This naturally gives rise to the second parameter regime described above, where $\tilde{c}$ has a constant number of erasures.

Locally recoverable (or repairable) codes (LRCs)~\cite{HCL13,GHSY12,xoring} were introduced to deal with this setting.  The guarantee of an LRC\footnote{In some works, the guarantee holds for information symbols only, rather than for all codeword symbols; we stick with all symbols here for simplicity of exposition.} with locality $Q$ is that for any $i \in \inset{1,\ldots,n}$, the $i$'th symbol of the codeword can be determined from a set of at most $Q$ other symbols. 
There has been a great deal of work recently aimed at pinning down the trade-offs between rate, distance, and the locality parameter $Q$ in LRCs. 
At this point, we have constructions which have optimal trade-offs between these parameters, as well as reasonably small alphabet sizes~\cite{TB14}.  However, there are still many open questions; a major question is how to handle a small number of erasures, rather than a single erasure.
This may result from either multiple node failures, or from ``hot" data being overloaded with requests.  
There are several approaches in the literature, but the approach relevant to this work is the study of \em multiple disjoint repair groups.\em 

\begin{definition} Given a code $\cC \subset \Sigma^N$, we say that a set $S \subset \inset{1,\ldots,N}$ is a \em repair group \em for $i \in \inset{1,\ldots,N}$ if $i \not\in S$, and if there is some function $g:\Sigma^{|S|} \to \Sigma$ so that $g(c|_S) = c_i$ for all $c \in \cC$.  That is, the codeword symbols indexed by $S$ uniquely determine the symbol indexed by $i$.
\end{definition}

\begin{definition}\label{def:sdrgp}
We say that $\cC$ has the $s$-Disjoint-Repair-Group Property ($s$-DRGP) if for every $i \in \inset{1,\ldots,N}$, there are $s$ disjoint repair groups $S_1^{(i)},\ldots,S_s^{(i)}$ for $i$.
\end{definition}

In the context of LRCs, the parameter $s$ is called the \em availability \em of the code.
An LRC with availability $s$ is not exactly the same as a code with the $s$-DRGP (the difference is that, in Definition~\ref{def:sdrgp}, there is no mention of the size $Q$ of the repair groups), but it turns out to be deeply related; it is also directly related to other notions of locality in distributed storage (like batch codes), as well as in cryptography (like PIR codes).  We will review some of this work below, and we point the reader to~\cite{Ska16} for a survey of batch codes, PIR codes, and their connections to LRCs and the $s$-DRGP. 

While originally motivated for small $s$, as we will see below, the $s$-DRGP is interesting (and has already been implicitly studied) for a wide range of $s$, from $O(1)$ to $\Omega(N)$.  For $s = o(N)$, we can hope for codes with very high rate, approaching $1$; the question is how fast we can hope for this rate to approach $1$.  More formally, if $K = \log_{|\Sigma|}|\cC|$, then the rate is $K/N$, and we are interested in how the gap $N-K$ behaves with $N$ and $s$.   We will refer to the quantity $N-K$ as the \em co-dimension \em of the code; when $\cC$ is linear (that is, when $\Sigma = \F$ is a finite field and $\cC \subseteq \F^N$ is a linear subspace), then this is indeed the co-dimension of $\cC$ in $\F^N$.  The main question we seek to address in this paper is the following.
 
 \begin{question}\label{q:main}
 	For a given $s$ and $N$, what is the smallest codimension $N - K$ of any code with the $s$-DGRP?  In particular, how does this quantity depend on $s$ and $N$?
 \end{question}
 
 We know a few things about Question~\ref{q:main}, which we survey below.  However, there are many things about this question which we still do not understand.  In particular, the dependence on $s$ is wide open, and this dependence on $s$ is the focus of the current work.  Below, we survey the state of Question~\ref{q:main} both from the LRC end (when $s$ is small) and the LCC end (when $s$ is large).
 
 \paragraph{The $s$-DRGP when $s$ is small.}
In~\cite{WZ14,RPDV14,TB_mult_rec,TBF16}, the $s$-DRGP was explicitly considered, with a focus on small $s$ ($s=2$ is of particular interest).  In those works, some bounds on the rate and distance of codes with the $s$-DRGP were derived (some of them in terms of the locality $Q$).  However, for larger $s$, these bounds degrade.
More precisely, \cite{WZ14,RPDV14} establish bounds on $N - K$ in terms of $Q,s$, and the distance of the code, but as $s$ grows these are not much stronger than the Singleton bound.
The results of~\cite{TB_mult_rec,TBF16} give an upper bound on the rate of a code in terms of $Q$ and $s$.  One corollary is that the rate satisfies $K/N \leq (s+1)^{-1/Q}$; if we are after high-rate codes, this implies that we must take $Q  = \Omega( \ln(s + 1) )$, and this implies that the codimension $N - K$ must be at least $\Omega(N \ln(s) / Q)$.  

A similar notion to the $s$-DRGP was introduced in~\cite{codedpir}, with the application of \em Private Information Retrieval \em (PIR).  PIR schemes are an important primitive in cryptography, and they have long been linked to constant-query LCCs.  In~\cite{codedpir}, PIR was also shown to be related to the $s$-DRGP.  The work~\cite{codedpir} introduces \em PIR codes, \em which enable PIR schemes with much less storage overhead.  It turns out that the requirement for PIR codes is very similar to the $s$-DRGP.\footnote{The only difference is that PIR codes only need to recover information symbols, but possibly with non-systematic encoding.}  

In the context of PIR codes~\cite{codedpir,BE16}, there are constructions of $s$-DRGP codes with $N - K \leq O(s \sqrt{N})$.  For $s=2$, this is known to be tight, and there is a matching lower bound~\cite{vardynote}.  However, it seems difficult to use this lower bound technique to prove a stronger lower bound when $s$ is larger (possibly growing with $N$). 

\paragraph{The $s$-DRGP when $s$ is large.}  
As we saw above, when $s$ is small then the $s$-DRGP is intimately related to LRCs, PIR codes and batch codes.
On the other end of the spectrum, when $s$ is large (say, $\Omega(N)$ or $\Omega(N^{1-\epsilon})$) then it is related to LCCs.

When $s = \Omega(N)$, then the $s$-DRGP  is in fact \em equivalent \em to a constant-query LCC (that is, an LCC as described above, where the number of queries to $\tilde{c}$ is $O(1)$).  
The fact that the $\Omega(N)$-DRGP implies a 
constant-query LCC is straightforward: the correction algorithm to recover
$c_i$ is to choose a random $j$ in $\inset{1,\ldots,s}$ and use the repair
group $S^{(i)}_j$ to recover $c_i$.  Since in expectation the size of $S^{(i)}_j$ is
constant, we can restrict our attention only to the constant-sized repair groups.  Then, with some constant probability none of the indices in $S^{(i)}_j$ will be
corrupted, and this success probability can be amplified by independent
repetitions.  The converse is also true~\cite{KT00,Woo10}, and any constant-query LCC has the $s$-DRGP for $s = \Omega(n)$; in fact, this connection is one of the few ways we know how to get lower bounds on LCCs.

When $s$ is large, but not as large as $\Omega(N)$, there is still a tight relationship with LCCs.   By now we know of several high-rate ($(1 - \alpha)$, for any constant $\alpha$) LCCs with query complexity $Q = N^\epsilon$ for any $\eps > 0$~\cite{KSY11, GKS13, HOW13} or even $Q = N^{o(1)}$~\cite{KMRS16}.
It is easy to see\footnote{Indeed, suppose that $\cC$ is an LCC with query complexity $Q$ and error tolerance $\delta$, and let $s = \delta N/Q$.
In order to obtain $s$ disjoint repair groups for a symbol $c_i$ from the LCC guarantee, we proceed as follows.  First, we make one (randomized) set of queries to $c$; this gives the first repair group.  Continuing inductively, assume we have found $t \leq s$ disjoint repair groups already, covering a total of at most $tQ < \delta N$ symbols.  To get the $t+1$'st set of queries, we again choose at random as per the LCC requirement.  These queries may not be disjoint from the previous queries, but the LCC guarantee can handle errors (and hence erasures) in up to $\delta N$ positions, so it suffices to query the points which have not been already queried, and treat the already-queried points as unavailable.  We repeat this process until $t$ reaches $s = \delta N/Q$.
}
that any LCC with query complexity $Q$ has the $s$-DGRP for $s = \Omega(N/Q)$.  Thus, these codes immediately imply high-rate $s$-DRGP codes with $s = \Omega(N^{1-\epsilon})$ or even larger.  (See also~\cite{AY17}).
Conversely, 
the techniques of~\cite{HOW13,KMRS16} show how to take high-rate linear codes with the $s$-DGRP for $s = \Omega(N^{1 - \epsilon})$ and produce high-rate LCCs with query complexity $O(N^{\epsilon'})$ (for a different constant $\epsilon'$).

These relationships provide some bounds on the codimension $N - K$ in terms of $s$: from existing lower bounds on constant-query LCCs~\cite{Woo10}, we know that any code with the $s$-DGRP and $s = \Omega(N)$ must have vanishing rate.  On the other hand from high-rate LCCs, there exist $s$-DGRP codes with $s = \Omega(N^{1 - \epsilon})$ and with high rate.  
However, these techniques do not immediately given anything better than high (constant) rate, while in Question~\ref{q:main} we are interested in precisely controlling the co-dimension $N - K$.

\paragraph{The $s$-DGRP when $s$ is intermediate.}  
The fact that the $s$-DRGP interpolates between the LRC setting for small $s$ and the LCC setting for large $s$ motivates the question of the $s$-DGRP for intemediate $s$, say $s = \log(N)$ or $s = N^c$ for $c < 1/2$.  Our goal is to understand the answer to Question~\ref{q:main} for intermediate $s$.

We have only a few data points to answer this question.  As mentioned above, the constructions of \cite{codedpir,BE16} show that there are codes with $N - K \leq s\sqrt{N}$ for $s \leq \sqrt{N}$.  However, the best general lower bounds that we have~\cite{vardynote,TBF16} can only establish 
\[N - K \geq \max \inset{ \sqrt{2N} , N - \frac{N}{(s+1)^{1/Q} }. }\]
Above, we recall that $Q$ is a parameter bounding the size of the repair groups; 
in order for the second term above (from~\cite{TBF16}) to be $o(N)$, we require $Q \gg \ln(s+1)$; in this case, the second bound on the codimension reads $N - K \geq \Omega(N\ln(s)/Q)$.   As the size of the repair groups $Q$ may in general be as large as $N/s$, in our setting this second bound gives better dependence on $s$, but worse dependence on $N$.  

The upper bound of $s \sqrt{N}$ is not tight, at least for large $s$.  For $s = \sqrt{N}$, there are several classical constructions which have the $s$-DRGP and with $N - K = \Theta( N^{\log_4(3)})$; for example, this includes affine geometry codes and/or codes constructed from difference sets (see~\cite{assmus1998}, \cite{LC04}, or \cite{GKS13}---we will also recover these in Corollary~\ref{cor:implicit}).  Notice that this is much better than the upper bound of $N - K \leq s\sqrt{N}$, which for $s = \sqrt{N}$ would be trivial.

However, other than these codes, before this work we did not know of any constructions for $s \ll \sqrt{N}$ which beat the bounds in \cite{codedpir, BE16} of $N - K \leq s\sqrt{N}$.\footnote{We note that there have been some works in the intermediate-$s$ parameter regime which can obtain excellent locality $Q$ but are not directly relevant for Question~\ref{q:main}.  In particular, the work of~\cite{RPDV14} gives a construction of $s$-DRGP codes with $s = \Theta(K^{1/3 - \eps})$ and $Q = \Theta(K^{1/3})$ for arbitarily small constant $\eps$; while this work obtains a smaller $Q$ than we will eventially obtain (our results will have $Q \sim \sqrt{N}$), they are only able to establish high (constant) rate codes, and thus do not yield tight bounds on the co-dimension.  	
	The work of \cite{AD14} gives constructions of high-rate fountain codes which have $s,Q = \Theta(\log(N))$.  As these are rateless codes, again they are not directly relevant to Question~\ref{q:main}.}   One of the main contributions of this work is to give a construction with $s = N^{1/4}$, which achieves codimension $N - K = N^{0.714}$.  Notice that the bound of $s \sqrt{N}$ would be $N^{0.75}$ in this case, so this is a substantial improvement.  We remark that we do not believe that our construction is optimal, and unfortunately we don't have any deep insight about the constant $0.714$.  Rather, we stress that the point of this work is to ( a) highlight the fact that the $s \sqrt{N}$ bound can be beaten for $s \ll \sqrt{N}$, and (b) highlight our techniques, which we believe may be of independent interest.  We discuss these in the next section.


\subsection{Lifted codes, and our construction}

Our construction is based on the \em lifted codes \em of Guo, Kopparty and Sudan~\cite{GKS13}.  The original motivation for lifted codes was to construct high-rate LCCs, as described above.  However, since then they have found several other uses, for example list-decoding and local-list-decoding~\cite{GK14}.
The codes are based on multivariate polynomials, and we describe them below.

Suppose that $\mathcal{F} \subseteq \F_q[X,Y]$ is a collection of bivariate polynomials over a finite field $\F_q$ of order $q$.  This collection naturally gives rise to a code $\cC \subseteq \F^{q^2}$:
\begin{equation}
\label{eq:polycode}
\cC= \inset{ \langle P(x,y) \rangle_{(x,y) \in \F_q^2 } \suchthat P \in \mathcal{F}}.
\end{equation}
Above,  we assume some fixed order on the elements of $\F_q^2$, and by $\langle P(x,y) \rangle_{(x,y) \in \F_q^2}$, we mean the vector in $\F_q^{q^2}$ whose entries are the evaluations of $P$ in this prescribed order.
For example, a bivariate Reed-Muller code is formed by taking $\mathcal{F}$ to be the set of all polynomials of total degree at most $d$.  

One nice property of Reed-Muller codes is their locality.  More precisely, suppose that $P(X,Y)$ is a bivariate polynomial over $\F_q$ of total degree at most $d$.  For an affine line in $\F_q^2$, parameterized as $L(T) = ( \alpha T + \beta, \gamma T + \delta)$, we can consider the \em restriction \em $P|_L$ of $P$ to $L$, given by
\[ P|_L(T) := P( \alpha T + \beta, \gamma T + \delta) \mod T^q - T, \]
where we think of the above as a polynomial of degree at most $q-1$.
It is not hard to see that if $P$ has total degree at most $d$, then $P|_L(T)$ also has degree at most $d$; in other words, it is a univariate Reed-Solomon codeword.  This property---that the restriction of any codeword to a line is itself a codeword of another code---is extremely useful, and has been exploited in coding theory since Reed's majority logic decoder in the 1950's~\cite{R54}.  
A natural question is whether or not there exist any bivariate polynomials $P(X,Y)$ \em other \em than those of total degree at most $d$ which have this property.  That is, are there polynomials which have high degree, but whose restrictions to lines are always low-degree?  In many settings (for example, over the reals, or over prime fields, or over fields that are large compared to the degrees of the polynomials) the answer is no.  However, the insight of~\cite{GKS13} is that there are settings---high degree polynomials over small-characteristic fields---for which the answer is yes.

This motivates the definition of \em lifted codes, \em which are multivariate polynomial evaluation codes, all of whose restrictions to lines lie in some other base code.  Guo, Kopparty and Sudan showed that, in the case above, not only do these codes exist, but in fact they may have rate much higher than the corresponding Reed-Muller code.

Lifted codes very naturally give rise to codes with the $s$-DRGP.  Indeed, consider the bivariate example above, with $d = q-2$.  That is, $\cC$ is the set of codewords arising from evaluations of functions $P$ that have the property that for all lines $L: \F_q \to \F_q^2$, $\deg(P|_L) \leq q-2$.  The restrictions then lie in the parity-check code: we always have $\sum_{t \in \F_q} P|_L(t) = 0$.  Thus, for every coordinate of a codeword in $\cC$---which  corresponds to an evaluation point $(x,y) \in \F_q^2$---there are $q$ disjoint repair groups for this symbol, corresponding to the $q$ affine lines through $(x,y)$.

However, it's not obvious how to use these codes to obtain the $s$-DRGP for $s \ll \sqrt{N}$; increasing the number of variables causes $s$ to grow, and this is the approach taken in~\cite{GKS13} to obtain high-rate LCCs.
Since we are after smaller $s$, we take a different approach.  
We stick with bivariate codes, but instead of requiring that the functions $P \in \mathcal{F}$ restrict to low-degree polynomials on \em all \em affine lines $L$, we make this requirement only for \em some \em lines.  This allows us to achieve the $s$-DRGP (if there are $s$ lines through each point), while still being able to control the rate.

We hope that our construction---and the machinery we develop to get a handle on it---may be useful more generally.    In the next section, we will set up our notation and give an outline of this approach, after a brief review of the notation we will use throughout the paper.

\paragraph{Outline.}
Next, in Section~\ref{sec:tech}, we define \em partially lifted codes, \em and give a technical overview of our approach.  This approach consists of two parts.  The first is a general framework for understanding the dimension of partially lifted codes of a certain form, which we then discuss more in Section~\ref{sec:framework}.  The second part is to instantiate this framework, which we do in Section~\ref{sec:constructions}.  The bulk of the work, in Section~\ref{sec:mainconstruction} is devoted to analyzing a particular construction which will give rise to that $s$-DRGP code with $s = N^{1/4}$ described above.    In Section~\ref{sec:variants} we mention a few other ways of constructing codes within this framework that seem promising.

\section{Technical Overview}\label{sec:tech}

In this section, we give a high-level overview of our construction and approach.  We begin with some basic definitions and notation.
\subsection{Notation and basic definitions}
We study linear codes $\cC \subseteq \F_q^N$ of block length $N$ over an alphabet of size $q$.  We will always assume that $\F_q$ has characteristic $2$, and write $q = 2^\ell$.  (We note that this is not strictly necessary for our techniques to apply---the important thing is only that the field is of relatively small characteristic---but it simplifies the analysis, and so we work in this special case).  

The specific codes $\cC$ that we consider are \em polynomial evaluation codes. \em  Formally, let $\mathcal{F}$ be a collection of $m$-variate polynomials over $\F_q$.  
Letting $N = q^m$, we may identify $\mathcal{F}$ with a code $\cC \subseteq \F_q^{N}$ as in \eqref{eq:polycode}; we assume that there is some fixed ordering on the elements of $\F_q^m$ to make this well-defined.  For a polynomial $P \in \F_q[X_1,\ldots,X_m]$, we write its corresponding codeword as
\[ \eval(P) = \langle P(x_1,\ldots,x_m) \rangle_{(x_1,\ldots,x_m)\in \F_q^m} \in \cC. \]
We will only focus on $m = 1,2$, as we consider the restriction of bivariate polynomial codes to lines, which results in univariate polynomial codes.  Formally, a (parameterization of an) \em affine line \em is a map $L: \F_q \to \F_q^2,$ of the form
\[ L(T) = (\alpha T + \beta, \gamma T + \delta )\]
for $\alpha,\beta,\gamma,\delta \in \F_q$. 
We say that two parameterizations $L, L'$ are \em equivalent \em if the result in the same line as a set:
\[ \inset{ L(t) \suchthat t \in \F_q } = \inset{L'(t) \suchthat t \in \F_q }. \]
We denote the restriction of a polynomial $P \in \F_q[X,Y]$ to $L$ by $P|_L$:
\begin{definition}
	For a line $L : \F_q \to \F_q^2$ with $L(T) = (L_1(T), L_2(T))$, and a polynomial $P : \F_q^2 \to \F_q$, we define the \textit{restriction} of $P$ on $L$, denoted $P\vert_L : \F_q \to \F_q$, to be the unique polynomial of degree at most $q-1$ so that $P|_L(T) = P(L_1(T), L_2(T))$.
\end{definition}
We note that the definition above makes sense, because all functions $f:\F_q \to \F_q$ can be written as polynomials of degree at most $q-1$ over $\F_q$; in this case, we have $P|_L(T) = P(L_1(T), L_2(T)) \mod (T^q - T)$.  

Finally, we'll need some tools for reasoning about integers and their binary expansions.
\begin{definition}\label{def:B}
	Let $m < q$ be a positive integer. If $m = \sum_{i = 0}^{\l -1 }m_i2^i$, where $m_i \in \{0, 1\}$, then we let $B(m) = \{i \in \{0, ..., \l - 1 \}\mid m_i = 1\}$. That is, $B(m)$ is the set of indices where the binary expansion of $m$ has a $1$.
\end{definition}
We say that an integer $m$ lies in the \em $2$-shadow \em of another integer $n$ if $B(m) \subseteq B(n)$:
\begin{definition}\label{def:shadow}
	For any two integers $m, n < q$, we say that $m$ lies in the $2$-shadow of $n$, denoted $m \leq_2 n$, if $B(m) \subseteq B(n)$. Equivalently, letting $m =  \sum_{i = 0}^{\l - 1}m_i 2^i$ and $n = \sum_{i = 0}^{\l - 1}n_i 2^i$, we write $m \leq_2 n$ if for all $i \in \{0, ..., \l - 1\}$, whenever $m_i = 1$ then also $n_i = 1$.
\end{definition}

The reason that we are interested in $2$-shadows is because of Lucas' Theorem, which characterizes when binomial coefficients are even or odd.\footnote{Lucas' Theorem holds more generally for any prime $p$, but in this work we are only concerned with $p=2$, and so we state this special case here.}
\begin{theorem}[Lucas' Theorem]\label{thm:lucas}
	For any $m, n \in \mathbb{Z}$, ${m \choose n} \equiv 0 \mod{2}$ exactly when $m \not\leq_2 n$. 
\end{theorem}
Finally, for integers $a,b,s$, we will say $a \equiv_s b$ if $a$ is equal to $b$ modulo $s$.  For a positive integer $n$, we use $[n]$ to denote the set $[n] = \inset{0,\ldots, n-1}$.

\subsection{Partially lifted codes}
With the preliminaries out of the way, we proceed with a description of our construction and techniques.
As alluded to above, our codes will be bivariate polynomial codes, which are ``partial lifts" of parity check codes.

\begin{definition}\label{def:partiallift}
	Let $\mathcal{F}_0 \subseteq \F_q[T]$ be a collection of univariate polynomials, and let $\mathcal{L}$ be a collection of parameterizations of affine lines $L:\F_q \to \F_q^2$.  We define the \em partial lift \em of $\mathcal{F}_0$ with respect to $\mathcal{L}$ to be the set
	\[ \mathcal{F} = \inset{P\in \F_q[X,Y] \suchthat \forall P \in \mathcal{F}, \forall L \in \mathcal{L}, P|_L \in \mathcal{F}_0 }. \]	
\end{definition}
We make a few remarks about Definition~\ref{def:partiallift} before proceeding.
\begin{remark}[Equivalent lines]
	We remark that the definition above allows $\mathcal{L}$ to be a collection of \em parameterizations \em of lines.  A priori, it is possible that equivalent parameterizations may behave very differently with respect to $\mathcal{F}_0$, and it is also possible to include
	several equivalent parameterizations in $\mathcal{L}$.
	 In this work, $\mathcal{F}_0$ will always be affine-invariant (in particular, it will just be the set of polynomials of degree strictly less than $q-1$), and so if $L$ and $L'$ equivalent, then $P|_L \in \mathcal{F}_0$ if and only if $P|_{L'} \in \mathcal{F}_0$.   Thus, these issues won't be important for this work.
\end{remark}
\begin{remark}[Why only bivariate lifts?]
This definition works just as well for $m$-variate partial lifts, and we hope that further study will explore this direction.  However, as all of our results are for bivariate codes, we will stick to the bivariate case to avoid having to introduce another parameter.
\end{remark}

Let $\mathcal{F}_0 := \inset{P \in \F_q[X], \deg(P) < q-1 }$.  Then it is not hard to see that the code $\cC_0 = \inset{ \eval(P) \suchthat P \in \mathcal{F}_0}$ is just the parity-check code,
\[ \cC_0 = \inset{c \in \F_q^q \suchthat \sum_{i =1}^q c_i = 0 } .\]
Indeed, for any $d < q-1$, we have $\sum_{x \in \F_q} x^d = 0$.

We will construct codes with the $s$-DRGP by considering codes that are partial lifts of $\mathcal{F}_0$.  We first observe that such codes, with an appropriate set of lines $\mathcal{L}$, will have the $s$-DRGP.  Indeed, suppose we wish to recover a particular symbol, given by $P(x,y)$ for $(x,y) \in \F_q^2$.
Let $L^{(1)}, \ldots, L^{(s)} \in \mathcal{L}$ be $s$ distinct (non-equivalent) lines that pass through $(x,y)$; say they are parameterized so that $L^{(j)}(0) = (x,y)$.  Then the $s$ disjoint repair groups are the sets indices corresponding to 
\[ S_j := \{ L^{(j)}(t) : t \in \F_q \setminus \inset{0} \}.\]
For any $P$ in the partial lift of $\mathcal{F}_0$, we have
\[ P|_L(0) = \sum_{t \in \F_q \setminus \{0\}} P|_L(t),\]
which means that
\[ P(x,y) = \sum_{(a,b) \in S_j}  P(a,b).  \]
That is, $P(x,y)$ can be recovered from the coordinates of $\eval(P)$ indexed by $S_j$, as desired.  Finally we observe that the $S_j$ are all disjoint, as the lines are all distinct, and intersect only at $(x,y)$.
We summarize the above discussion in the following observation.
\begin{observation}\label{obs:keyobs}
	Suppose that $\mathcal{F}_0 = \inset{P \in \F_q[T] \suchthat \deg(P) < q-1}$, and let $\mathcal{L}$ be any collection of parameterizations of affine lines so that every point in $\F_q^2$ is contained in at least $s$ non-equivalent elements of $\mathcal{L}$.  Let  $\mathcal{F}$  be the bivariate partial lift of $\mathcal{F}_0$ with respect to $\mathcal{L}$.  Then the code $\mathcal{C}\subseteq \F_q^{q^2}$ corresponding to $\mathcal{F}$ is a linear code with the $s$-DRGP.
\end{observation}

To save on notation later, we say that a polynomial $P:\F_q^2 \to \F_q$ \em restricts nicely \em on a line $L: \F_q \to \F_q^2$ if $P|_L$ has degree strictly less than $q-1$. 
Thus, to define our construction, we have to define the collection $\mathcal{L}$ of lines used in Definition~\ref{def:partiallift}.  We will actually develop a framework that can handle a family of such collections, but for intuition in this section, let us just consider lines $L(T) = (T, \alpha T + \beta)$ where $\alpha$ lives in a multiplicative subgroup $G_s$ of $\F_q^*$ of size $s$, and $\beta \in \F_q$.  That is, we are essentially restricting the slope of the lines to lie in a multiplicative subgroup.  It is not hard to see that every point $(x,y) \in \F_q^2$ has $s$ non-equivalent lines in $\mathcal{L}$ that pass through it.

Following Observation~\ref{obs:keyobs}, the resulting code will immediately have the $s$-DRGP.  The only question is, what is the rate of this code?  Equivalently, we want to know:
\begin{question}\label{q:howmany}
	How many polynomials $P \in \F_q[X,Y]$ have $\deg(P|_L) < q - 1$  for all $L \in \mathcal{L}$, where $\mathcal{L}$ is as described above?
\end{question}
In \cite{GKS13}, Guo, Kopparty and Sudan develop some machinery for answering this question when $\mathcal{L}$ is the set of all  affine lines.  What they show in that work is that in fact the (fully) lifted code is affine-invariant, and is equal to the span of the monomials $P(X,Y) = X^aY^b$ so that $\deg(P|_L) < q - 1$ for all affine lines $L$.  
We might first hope that this is the case for partial lifts---but then upon reflection we would immediately retract this hope, because it turns out that we do not get any more monomials this way: Theorem~\ref{thm:gks} establishes that if a monomial restricts nicely on even one line of the form $(T, \alpha T + \beta)$ (for nonzero $\alpha,\beta$), then in fact it restricts nicely on \em all \em such lines.  
In fact, the partial lift is not in general affine-invariant, and this is precisely where we are able to make progress.   More precisely, there may be polynomials $P(X,Y)$ of the form 
\begin{equation}\label{eq:nicebinomial} 
P(X,Y) = X^{a_1}Y^{b_1} + X^{a_2}Y^{b_2}
\end{equation}
 which are contained in the partial lift $\mathcal{F}$, but so that $X^{a_1}Y^{b_1}, X^{a_2}Y^{b_2} \not\in \mathcal{F}$.   This gives us many more polynomials to use in a basis for $\mathcal{F}$ than just the relevant monomials, and allows us to construct families $\mathcal{F}$ of larger dimension.
 
 \begin{remark}[Breaking affine invariance] We emphasize that breaking affine-invariance is a key departure from~\cite{GKS13}.  In some sense, it is not surprising that we are able to make progress by doing this: the assumption of affine-invariance is one way to prove \em lower bounds \em on locality~\cite{BS11,BG16}.   This is also where our techniques diverge from those of \cite{GKS13}.   Because of their characterization of affine-invariant codes, that work focused on understanding the dimension of the relevant set of monomials.  This is not sufficient for us, and so to get a handle on the dimension of our constructions, we must study more complicated polynomials.  This may seem daunting, but we show---perhaps surprisingly---that one can make a great deal of progress by considering only the additional ``more complicated" polynomials of the form~\eqref{eq:nicebinomial},  which are arguably the simplest of the ``more complicated" polynomials.
\end{remark}

In order to obtain a lower bound on the dimension of $\mathcal{F}$, our strategy get a handle on the dimension of the space of these binomials \eqref{eq:nicebinomial}.  If we can show that there are many linearly independent such binomials, then the answer to Question~\ref{q:howmany} must be ``lots."

Following this strategy, we examine binomials of the form \eqref{eq:nicebinomial}, and we ask, for which $a_1,b_1,a_2,b_2$ and which $L(T) = (T, \alpha T + \beta)$ does $P(X,Y)$ restrict nicely?  
Our main tool is Lucas's Theorem (Theorem~\ref{thm:lucas}), which was also used in \cite{GKS13}.  To see why this is useful, consider the restriction of a monomial $P(X,Y) = X^aY^b$ to a line $L(T) = (T, \alpha T + \beta)$.  We obtain
\[ P|_L(T) = T^a \inparen{ \alpha T + \beta }^b = \sum_{i \leq b} {b \choose i} \alpha^i \beta^{b-i} T^{a + i}. \]
Above, the binomial coefficient ${b \choose j}$ is shorthand for the sum of $1$ with itself ${ b \choose j}$ times.  Thus, in a field of characteristic $2$, this is either equal to $1$ or equal to $0$; Lucas's theorem tells us which it is.   This means that our question reduces to asking, when does the coefficient of $T^{q-1}$ vanish?  The above gives us an expression for this coefficient, and allows us to compute an answer, in terms of the binary expansions of $a$ and $b$.

So far, this is precisely the approach of \cite{GKS13}.  From here, we turn to the binomials of the form \eqref{eq:nicebinomial}.  When do these restrict nicely?  As above, we may compute the coefficient of the $T^{q-1}$ term and examine it.  Fortunately, when the set of lines $\mathcal{L}$ is chosen as above, the number of linearly independent binomials that restrict nicely ends up having a nice expression, in terms of the number of non-empty equivalence classes of a particular relation defined by the binary expansion of the numbers $1, \ldots, q-1$; this is our main technical theorem (Theorem~\ref{thm:equiv}, which is proved in Section~\ref{sec:relax}). 

The approach of Section~\ref{sec:relax} holds for more general families than the $\mathcal{L}$ described above; instead of taking $\alpha$ in a multiplicative subgroup of $\F_q^*$, we may alternately restrict $\beta$, or restrict both.   However, numerical calculations indicated that the choice above (where $\alpha$ is in a multiplicative subgroup of order $s$) is a good one, so for our construction we make this choice and we focus on that for our formal analysis in Section~\ref{sec:constructions}. 

In order to get our final construction and obtain the results advertised above, it suffices to count these equivalence classes.
For the result advertised in the introduction, we choose the order of the multiplicative subgroup to be $s = 2^{\ell/2} - 1  = \sqrt{q} - 1$.  
Then, we use an inductive argument in Section~\ref{sec:constructions} to count the resulting equivalence classes, obtaining the bounds advertised above.
More precisely, we obtain the following theorem.

\begin{theorem}\label{thm:main}
	Suppose that $q = 2^\ell$ for even $\ell$, and let $N = q^2 - 1$.  There is a linear code $\cC$ over $\F_q$ of length $N$ and dimension 
	\[ K \geq N - O(N^{.714}) \]
	which has the $s$-DRGP for $s = \sqrt{q} - 2 =  (N + 1)^{1/4} - 1$.
\end{theorem}

\begin{remark}[Puncturing at the origin]
We note that the statement of the theorem differs slightly from the informal description above; in our analysis, we will puncture the origin, and ignore lines that go through the origin; that is, our codes will have length $q^2 -1$, rather than $q^2$, and the number of lines through every point will be $s-1$, rather than $s$, as it makes the calculations somewhat easier and does not substantially change the results.
\end{remark}

\subsection{Discussion and open questions}
Before we dive into the technical details in Section~\ref{sec:framework}, we close the front matter with some discussion of open questions left by our work and our approach. 
We view the study of the $s$-DRGP for intermediate $s$ to be an important step in understanding locality in general, since the $s$-DRGP nicely interpolates between the two extremes of LRCs and LCCs.   When $s=2$, we completely understand the answer to Question~\ref{q:main}.  However, by the time $s$ reaches $\Omega(N)$, this becomes a question about the best rate of constant-query LCCs, which is a notoriously hard open problem.  It is our hope that by better understanding the $s$-DRGP, we can make progress on these very difficult questions.  

The main question left by our work is Question~\ref{q:main}, which we do not answer.  What is the correct dependence on $s$ in the codimension of codes with the $s$-DRGP?  We have shown that it is not $s\sqrt{N}$, even for $s \ll \sqrt{N}$.  However, we have no reason to believe that our construction is optimal.

Our work also raises questions about partially lifted codes.  These do not seem to have been studied before.  The most immediate question arising from our work is to improve or generalize our approach; in particular, is our analysis tight?
Our approach proceeds by counting the binomials of the form \eqref{eq:nicebinomial}.  This is in principle lossy, but empirical simulations suggest that at least in the setting of Theorem~\ref{thm:main}, this approach is basically tight.  Are there situations in which this is not tight?  Or can we prove that it is tight in any situation?   Finally, are there other uses of partially lifted codes?  As with lifted codes, we hope that these prove useful in a variety of settings.

\section{Framework}\label{sec:framework}
As discussed in the previous section, the proof of Theorem~\ref{thm:main} is based on the partially lifted codes of Definition~\ref{def:partiallift}.
In this section, we lay out the partially lifted codes we consider, as well as the basic tools we need to analyize them.
As before, we say that a polynomial $P:\F_q^2 \to \F_q$ \em restricts nicely \em to a line $L: \F_q \to \F_q^2$ if $P|_L$ has degree strictly less than $q-1$.
We will consider partial lifts of the parity-check code with respect to a collection of affine lines $\mathcal{L}$; reasoning about the rate of this code will amount to reasoning about the polynomials which restrict nicely to lines in $\mathcal{L}$.

To ease the computations, we will form our family $\mathcal{L}$ out of lines that have a simple parameterization:
\begin{definition}
We say a line $L : \F \to \F^2$ is \em simple \em  if it can be written in the form $L(T) = (T, \a T + \b)$, with $\a, \b \neq 0$.
\end{definition}
Notice that this rules out lines through the origin.  At the end of the day, we will pucture our code at the origin to achieve our final result.  
Note also that no two simple parameterizations of lines are equivalent to each other (that is, they form distinct lines as sets), so as we go forward, we may apply Observation~\ref{obs:keyobs} without worry of the repair groups coinciding.

We consider a family of constructions, indexed by parameters $s$ and $t$, so that $s,t \divides q-1$.
This family will be the partial lift with respect to the following set of simple lines.
\begin{definition}
	\label{def:lst}
	Let $\mathcal{L}_{s,t}$ be the family of simple lines
\[ \mathcal{L}_{s,t} = \inset{ L(T) = (T , \a T + \b) \suchthat \alpha \in G_s, \beta \in G_t }.\]
\end{definition}

For the rest of the paper, we will study the following construction, for various choices of $s$ and $t$.x
\begin{construction}\label{cons:theconstruction}
		Suppose that $s,t \divides q-1$, and let $G_s, G_t \leq \F_q^*$ be multiplicative subgroups of $\F_q^*$ of orders $s$ and $t$, respectively.  That is, $G_s = \inset{x \in \F_q^* \suchthat x^s = 1 }$ and $G_t = \inset{x \in \F_q^* \suchthat x^t = 1 }$.
		Let $\mathcal{L}_{s,t}$ be as in Definition~\ref{def:lst}, and
		let $\mathcal{F}_0$ be the set of univariate polynomials of degree strictly less than $q-1$.  Define
		$\mathcal{F}_{s,t}$ to be the partial lift of $\mathcal{F}_0$ with respect to $\mathcal{L}_{s,t}.$
\end{construction}

Our main theorem, which we will prove in the rest of this section, is a characterization of the dimension of $\mathcal{F}_{s,t}$ as in Construction~\ref{cons:theconstruction}.  (We recall the definition of $\leq_2$ from Definition~\ref{def:shadow} above).
\begin{theorem}\label{thm:equiv}
	Suppose that $s,t \divides q-1$.  For nonnegative integers $i < s, j < t$, define 
	\[ e(s,t) = \inabs{\inset{ (i,j) \suchthat i < s, \text{ and } j < t, \text{ so that there is some } m ,n \in [q]^2 \text{ with }  m \equiv_s i, n \equiv_t j, \text{ and } n \leq_2 m}}.\]
	Then the dimension of $\mathcal{F}_{s,t} \subseteq \F_q[X, Y]$ is at least
	\[ \dim(\mathcal{F}_{s,t}) \geq q^2 - e(s,t). \]
\end{theorem}
Theorem~\ref{thm:equiv} may seem rather mysterious.  As we will see, the reason for the expression $e(s,t)$ is because it comes up in counting the number of binomials of the form \eqref{eq:nicebinomial} the restrict nicely on lines in $\mathcal{L}_{s,t}$.  
We'll re-state Theorem~\ref{thm:equiv} later as Theorem~\ref{thm:gendim} in Section~\ref{sec:relax}, after we have developed the notation to reason about $e(s,t)$, and we will prove it there.

The reason that Theorem~\ref{thm:equiv} is useful is that for some $s$ and $t$, it turns out to be possible to get a very tight handle on $e(s,t)$.  
This leads to the quantitative result in Theorem~\ref{thm:main}, which we will prove in Section~\ref{sec:constructions}.
For now, we focus on proving Theorem~\ref{thm:equiv}.  Our starting point is the work of \cite{GKS13}; we summarize the relevant points below in Section~\ref{sec:monomials}.

\subsection{Basic Setup: Lucas' Theorem and Monomials}\label{sec:monomials}
In \cite{GKS13}, Guo, Kopparty and Sudan give a characterization of lifted codes.  In our setting, their work shows that when the set $\mathcal{L}$ is the set of \em all \em affine lines, then the  lifted code $\mathcal{F}$ is affine invariant and in fact is equal to the span of the \em monomials \em  which restrict nicely.  In the case where the number of variables is large, or the base code $\mathcal{F}_0$ is more complicated than a parity-check code, \cite{GKS13} provides some bounds, but it seems quite difficult to get a tight characterization of these monomials.  However, for bivariate lifts of the parity-check code, it is actually possible to completely understand the situation, and this was essentially done in \cite{GKS13}.  We review their approach here.

First, we use Lucas' Theorem (Theorem~\ref{thm:lucas}) to characterize which monomials $X^aY^b$ restrict nicely to simple lines. Theorem $\ref{thm:gks}$ follows from the analysis in \cite{GKS13}; we provide the proof below for completeness.
\begin{theorem}\label{thm:gks}
Suppose $a + b < 2(q - 1)$ and let $P(X, Y) = X^aY^b$. Then for all simple lines $L(T) = (T, \a T + \b)$, $P\vert_L$ has degree $< q - 1$ if and only if $q - 1 - a \not\leq_2 b$. Further, if $q - 1 - a \leq_2 b$, then $P\vert_L$ is a degree $q - 1$ polynomial with leading coefficient $\a^{-a}\b^{b + a}$
\end{theorem}

\begin{proof}
Let $L(T) = (T, \alpha T + \beta)$ be a simple line, and let $P(X,Y) = X^aY^b.$
We compute the restriction of $P$ to $L$, 
and obtain
\begin{equation}
\label{eq:bigsum}
P\vert_L = T^a(\a T + \b)^b =  \sum_{i = 0}^{b}{b \choose i} \a^iT^{a + i}\b^{b - i},
\end{equation}
where in the above the binomial coefficient ${b \choose i}$ is shorthand for the sum of $1$ ${b \choose i}$ times in $\F_q$.  Since $q$ is of characteristic $2$, this is either $0$ or $1$.

Because $a + b < 2(q - 1)$, the only $i$ so that $T^{a + i} = T^{q - 1}$ is $i = q - 1 - a$. We thus compute the degree $q-1$ term as 
$${b \choose q - 1 - a} \a^{q - 1 - a}\b^{b - (q - 1 - a)}T^{q - 1} = {b \choose q - 1 - a} \a^{- a}\b^{b + a}T^{q - 1}.$$

To see when this term vanishes, we turn to Lucas' Theorem (Theorem \ref{thm:lucas}), which implies that ${b \choose q - 1 - a}$ is even exactly when $q - 1 - a \not\leq_2 b$.   Because $\F_q$ has characteristic 2, this means that ${b \choose q-1 -a} = 0$ in $\F_q$, and so the coefficient on $T^{q - 1}$ is zero whenever $q - 1 - a \not\leq_2 b$, as desired.

To see the second part of the theorem, suppose that $q - 1 - a \leq_2 b$, so Lucas' Theorem implies that
 ${b \choose q - 1 - a}$ is odd, and is equal to $1$ in $\F_q$.
In this case, we can see that the coefficient on $T^{q - 1}$ is $\a^{-a}\b^{b + a}$, as desired.
\end{proof}

Theorem~\ref{thm:gks} implies that whether a monomial $P(X,Y) = X^aY^b$ restricts nicely to a simple line $L$ is independent of the choice of $L$.
Thus it makes sense to consider this a property of the monomial itself.
\begin{definition}\label{def:good}
We say that a monomial $P(X, Y) = X^aY^b$ with $0 \leq a, b \leq q - 1$ is \em good \em if it restricts nicely on all simple lines. 
\end{definition}

\begin{remark}[The special case of $X^{q-1}Y^{q-1}$]\label{rem:special}
In Theorem \ref{thm:gks}, we required $a + b < 2(q - 1)$, which does not cover the monomial $P_*(X,Y) = X^{q-1}Y^{q-1}$.  However, in Definition \ref{def:good}, we allow $a = b = q - 1$, and in fact according to this definition $P_*(X,Y)$ is good.  Indeed, when we compute $(X^{q - 1}Y^{q - 1})\vert_L$ for any simple line $L(T) = (T, \a T + \b)$, there are two choices of $i$ in \eqref{eq:bigsum} that contribute to the $T^{q-1}$ term: $i = 0$ and $i = q - 1$. In particular, the coefficient on $T^{q - 1}$ is $${{q - 1} \choose 0}\a^0\b^{q - 1} +  {{q - 1} \choose {q - 1}}\a^{q - 1}\b^{0} = 1 + 1 = 0.$$ 
Since we consider only simple lines in our work, we will count $P_*(X,Y)$ as good, in addition to the monomials covered by Theorem~\ref{thm:gks}.

We note that there are lines $L$ which are not simple for which $\deg(P_*) = q-1,$ so it would not be included as ``good" in the analysis of~\cite{GKS13}.  (In their language, $P_*$ does not live in the lift of the degree set $\{0,\ldots,q-2\}$). 
\end{remark}

\begin{corollary}\label{cor:goodcount}
There are $q^2 - 3^\l + 1$ good monomials.
\end{corollary}

\begin{proof}
In order to count the number of good monomials, we notice that a monomial $P(X, Y) = X^aY^b$ with $a + b < 2(q - 1)$ is \textit{not} good when $q - 1 - a \leq_2 b$, i.e., $B(q - 1 - a) \subseteq B(b)$. Since $B(q - 1 - a)$ is exactly the compliment of $B(a)$, such a $P$ is not good if any only if $\overline{B(a)} \subseteq B(b)$.   Suppose that 
the binary expansion of $a$ is $a_{\l - 1}\cdots a_1 a_0$ and that the binary expansion of $b$ is $b_{\l - 1}\cdots b_1 b_0$, with $a_i, b_i \in \{0, 1\}$.

If $P$ is not good, then there are three options for each $i$: 
$a_i = 0$ and $b_i = 1$, or $a_i = 1$ and $b_i = 0$, or $a_i = b_i = 1$.  This yields $3^\ell$ monomials which are not good.  However, we must exclude the case where $a = b = q - 1$, because (as per Remark~\ref{rem:special}),  $X^{q - 1}Y^{q -1}$ is a good monomial. Thus there are $3^\l - 1$ total monomials which are not good.
Because there are $q^2$ monomials total, there are $q^2 - 3^\l + 1$ good monomials.
\end{proof}

At this point, we have recovered the codes of Theorem 1.2 in~\cite{GKS13}, up to the technicalities about simple lines vs. all lines.  Following Observation~\ref{obs:keyobs}, these codes have the $s$-DRGP for $s = q-1$; indeed, there are $q - 1$ simple lines through every non-zero point of $\F_q^2$.
The dimension of these codes is at least the number of monomials that they contain (indeed, all monomials are linearly independent), which by the above is at least $q^2 - 3^\ell + 1 = (N + 1) - (N+1)^{\log_4(3)} + 1$.
\begin{corollary}[Implicit in \cite{GKS13}]\label{cor:implicit}
	There are codes linear $\cC$ over $\F_q$ of length $N = q^2 - 1$  with dimension
	\[ K \geq N + 2 - (N+1)^{\log_4(3)} \]
	which have the $s$-DRGP for $s = q - 1 = \sqrt{N + 1} - 1$.
\end{corollary}
We remark that this just as easily produces codes of length $N = q^2$ with the $s$-DRGP for $s = q$ and dimension $K = N - N^{\log_4(3)}$, by allowing non-simple lines.  In what follows, simple lines will be much easier to work with, so we state things this way above for continuity (even though it looks a bit more messy).

We note that this recovers the results of one of the classical constructions of the $s$-DRGP for $s = \sqrt{N}$ mentioned in the introduction (and this is not an accident: these codes are in fact the same as affine geometry codes).  In the next section, we show how to use the relaxation to partial lifts in order to create codes with the $s$-DRGP for $s \ll \sqrt{N}$.

\subsection{Partially lifted codes}\label{sec:relax}
In this section we extend the analysis above to partial lifts.  The work of \cite{GKS13} characterizes the polynomials which restrict nicely on all lines
 $L : \F_q \to \F_q^2$: they show that this is exactly the span of the good monomials (except the special monomial $P_*$ of Remark~\ref{rem:special}, which restricts to degree lower than $q - 1$ only on \textit{simple} lines).  However, since our goal is to obtain codes with the $s$-DRGP for $s \ll \sqrt{N}$, increasing the dimension while decreasing $s$, we would like to allow for more polynomials.
 
  Thus, as in Definition~\ref{def:partiallift}, we will consider polynomials which restrict nicely only on some particular subset $\L$ of simple lines. We would like to find a subset $\L$ such that the space of polynomials which restrict nicely on all lines in $\L$ has large degree. Additionally, we would like to guarantee the $s$-DRGP by ensuring that, for every point $(x, y)$, there are many lines in $\L$ that pass through $(x, y)$. Relaxing requirements in this manner will allow us to get codes with good rate and locality trade-offs.

Theorem~\ref{thm:gks} shows that if a monomial restricts nicely on one simple line, it will restrict nicely on all simple lines. This means that in order to find a larger space of polynomials, we cannot only consider monomials. Towards this end, we will consider \em binomials \em of the form 
\begin{equation}\label{eq:nicebin2}
P(X,Y) = X^{a_1}Y^{b_1} + X^{a_2}Y^{b_2}.
\end{equation}
 That is, we will look only at binomials with both coefficients equal to 1.

We note that this ability to extend beyond monomials is possible crucially because our partially lifted codes are not affine-invariant.  While affine-invariance allowed \cite{GKS13} to get a beautiful characterization of (fully) lifted codes, it also greatly restricts the flexibility of these codes.   By breaking affine-invariance, we also break some of the rigidity of these constructions.  This is in some sense not surprising: affine invariance is often exploited in order to prove \em lower bounds \em on locality~\cite{BS11,BG16}. 

\subsubsection{Which binomials play nice with which lines?}

We would like to characterize which binomials of the form \eqref{eq:nicebin2} restrict nicely on which lines.   Unlike the case with monomials, now this will depend on the line as well as on the binomial.
When both individual terms in the binomial are good monomials, the binomial will certainly restrict nicely. However, if this is not the case, 
then the binomial could still restrict nicely, if the contributions to the leading coefficient of $P|_L$ from the two terms cancel with each other.
We start with a lemma characterizing the leading coefficient produced by restricting the sum of two not good monomials, and then state and prove
Lemma~\ref{lem:rest_bin}, which characterizes which binomials restrict nicely to which lines.

\begin{lemma}\label{lem:bin_coef}
Suppose that $P_1(X, Y) = X^{a_1}Y^{b_1}$ and $P_2(X, Y) = X^{a_2}Y^{b_2}$ are not good monomials.  Let  $L = (T, \a T + \b)$ be a simple line.  Then the coefficient of $T^{q - 1}$ in the restriction $(P_1 + P_2)\vert_L$ is $$\a^{-a_1}\b^{b_1 + a_1} + \a^{-a_2}\b^{b_2 + a_2}.$$
\end{lemma}

\begin{proof}
Because $P_1$ and $P_2$ are not good, Theorem \ref{thm:gks} implies that $P_1\vert_L$ has degree $q - 1$ with leading coefficient $\a^{-a_1}\b^{b_1 + a_1}$, and similarly $P_2\vert_L$ has degree $q - 1$ with leading coefficient $\a^{-a_2}\b^{b_2 + a_2}$. Because $(P_1 + P_2)\vert_L = P_1\vert_L + P_2\vert_L$, the coefficient of $T^{q - 1}$ in the restriction $(P_1 + P_2)\vert_L$ is in fact $\a^{-a_1}\b^{b_1 + a_1} + \a^{-a_2}\b^{b_2 + a_2}$, as desired.
\end{proof}

Lemma~\ref{lem:bin_coef} immediately implies the following characterization of when a binomial $P$ of the form~\eqref{eq:nicebin2} restricts nicely to a simple line $L$.

\begin{lemma}[Restricting Binomials]\label{lem:rest_bin}
Let $P(X, Y) = X^{a_1}Y^{b_1} + X^{a_2}Y^{b_2}$, and let $L = (T, \a T + \b)$ be a simple line.  Then $P$ restricts nicely to $L$  if and only if one of the following two conditions is met:

\begin{enumerate}
\item[(a)] both $X^{a_1}Y^{b_1}$ and $X^{a_2}Y^{b_2}$ are good, or
\item[(b)] neither $X^{a_1}Y^{b_1}$ nor $X^{a_2}Y^{b_2}$ are good, and \begin{equation}\label{eq:bin_coef}
\a^{a_2 - a_1} = \b^{b_2 - b_1 + a_2 - a_1}
\end{equation}
\end{enumerate}
\end{lemma}

\begin{proof}

First, we show that if $P$ and $L$ meet either of these conditions, then $P\vert_L$ has degree less than $q - 1$.
If $P$ and $L$ meet condition (a), then this follows immediately from the definition of a good monomial.
%
On the other hand, if $P$ and $L$ meet condition (b), 
then this follows from Lemma \ref{lem:bin_coef}, using the assumption that $\a, \b \neq 0$.
 If $\a^{a_1 - a_1} = \b^{b_2 - b_1 + a_2 - a_1}$, then $\a^{-a_1}\b^{b_1 + a_1} = \a^{-a_2}\b^{b_2 + a_2}.$  Because $\F_q$ has characteristic $2$, the coefficient on $T^{q-1}$ of $P|_L$ is
$$\a^{-a_1}\b^{b_1 + a_1} + \a^{-a_2}\b^{b_2 + a_2} = 0.$$ 

For the other direction, we show that if a binomial $X^{a_1}Y^{b_1} + X^{a_2}Y^{b_2}$ does not meet either condition on some line $L$, then the restriction  $P|_L$ will have a non-zero coefficient on $T^{q - 1}$. If exactly one of the terms of the binomial is a good monomial, the coefficient on $T^{q - 1}$ will be determined by the other term and cannot be zero. This leaves only the case where neither term restricts nicely along $L$, but $\a^{a_2 - a_1} \neq \b^{b_2 - b_1 + a_2 - a_1}$, which implies that the coefficient on $T^{q-1}$ in $P|_L$ is $\a^{-a_1}\b^{b_1 + a_1} + \a^{-a_2}\b^{b_2 + a_2} \neq 0$.
Thus, $P$ does not restrict nicely on $L$.
\end{proof}

	We are now primarily interested in how to satisfy \eqref{eq:bin_coef} in the second case of Lemma \ref{lem:rest_bin}. 
	We will do this by focusing on the special case\footnote{%
	One might ask, why this special case?  Why not consider
	 $\a^{a_1 - a_1} = \b^{b_2 - b_1 + a_2 - a_1} = c$ for some $c \in \F_q^* \setminus \inset{1}$?
	This divides the lines and binomials into equivalence classes based on the value of $c$, so that each class of binomials restrict nicely on every line in the corresponding class of lines. Unfortunately, it turns out that these classes are relatively small, and do not produce useful codes.  Thus we focus on the case where both sides are equal to $1$.
	}
	where $\a^{a_1 - a_1} = \b^{b_2 - b_1 + a_2 - a_1} = 1$. 
	We address this case with the next corollary.

\begin{corollary}\label{cor:st}
Let $s$ and $t$ divide $q - 1$, and let $G_s = \{x \in \F_q \suchthat x^s = 1\}$ and $G_t = \{x \in F_q \suchthat x^t = 1\}$. 
Let 
$$\Lst = \{(T, \a T + \b) \suchthat \a \in G_s, \b \in G_t\}$$
as in Definition~\ref{def:lst}.
Suppose that $P(X, Y) = X^{a_1}Y^{b_1} + X^{a_2}Y^{b_2}$ is a binomial so that neither term is good.  
Suppose that $a_1 \equiv a_2 \mod{s}$ and $a_1 + b_1 \equiv a_2 + b_2 \mod{t}$.
Then for all $L \in \Lst$, $P$ restricts nicely to $L$.
\end{corollary}

\begin{proof}
Let $L \in \Lst$ be a simple line with $L(T) = (T, \a T + \b)$. If $a_1 \equiv a_2 \mod{s}$ and $a_1 + b_1 \equiv a_2 + b_2 \mod{t}$, we know that $a_2 - a_1 = c_1s$ and $b_2 - b_1 + a_2 - a_1 = c_2 t$ for some integers $c_1$ and $c_2$. Then, because $\a \in G_s$ and $\b\in G_t$, we can see that $\a^{a_2 - a_1} = \a^{c_1s} = 1$, and similarly $\b^{b_2 - b_1 + a_2 - a_1} = \b^{c_2t} = 1$. Thus $\a^{a_2 - a_1} = \b^{b_2 - b_1 + a_2 - a_1}$, so $P\vert_L$ must have degree less than $q - 1$.
\end{proof}

Thus, a choice of $s$ and $t$ dividing $q-1$ produces a code by using $\Lst$ in Construction~\ref{cons:theconstruction}.
Each choice of $s$ and $t$ produces a different code, and by varying $s$ and $t$ we can vary the parameters of this code. This is the general framework for our construction, but we still must explore the dimension and the number of disjoint repair groups produced by different choices of $s$ and $t$.

\subsubsection{Dimension}

Given some choice of $s$ and $t$, we would like to understand dimension of the space of polynomials $\mathcal{F}_{s,t}$ which restrict nicely on all lines in $\Lst$. We will lower bound this dimension by building a linearly independent set $S \subseteq \mathcal{F}_{s,t}$ comprised of monomials and binomials. In order to construct $S$ and understand its size, we will need some more notation.

\begin{definition}\label{def:eij}
Let $i < s$ and $j < t$ be nonnegative integers. Define
$$E_{i, j} = \{(m,n) \in [q]^2 \suchthat m \equiv_s i, n \equiv_t j,  n \leq_2 m\}.$$
Further define $e(s, t)$ to be 
\[ e(s,t) = | \inset{(i,j) \suchthat E_{i,j} \neq \emptyset } | \]
the number of $(i, j)$ with $E_{i, j}$ nonempty, i.e., so that it is possible to find at least some $m \equiv_s i$ and $n \equiv_t j$ so that $n \leq_2 m$.
\end{definition}

\begin{definition}\label{def:potentialcancel}
We will call two monomials $P_1(X, Y) = X^{a_1}Y^{b_1}$ and $P_2(X, Y) = X^{a_2}Y^{b_2}$  \em potentially canceling \em  if $a_1 \equiv a_2 \mod{s}$ and $a_1 + b_1 \equiv a_2 - b_2 \mod{t}$. 
\end{definition}

That is, $P_1$ and $P_2$ are potentially canceling if the second hypothesis of Corollary~\ref{cor:st} holds.
However, we note that while Corollary~\ref{cor:st} applies only to binomials where neither term is good, Definition~\ref{def:potentialcancel} holds
 even if one or both of the monomials is good.  This explains the reason for the name \em potentially \em canceling: two potentially canceling monomials will cancel if neither is good, but may not cancel if one or both are good. This notion of potentially canceling gives us equivalence relation on monomials. In particular, we can divide monomials into equivalence classes $$M_{i, j} = \{ X^aY^b \mid a \equiv_s i, b + a \equiv_t j \},$$ so that any two monomials in $M_{i, j}$ are potentially canceling.  In order to get a handle on the binomials $P_1 + P_2$ which \em do \em cancel, we will first get a handle on the $M_{i,j}$, and how the good monomials are distributed between them.

Note that the equivalence classes $M_{i, j}$ divide up the $q^2$ total monomials into $st$ classes, and it is easy to characterize how many monomials are in each class. However, it is not necessarily so easy to describe how the \textit{good} monomials are distributed into these classes. Moreover, understanding this distribution will be critical to understanding the dimension of $\mathcal{F}_{s, t}$. Toward this end, we will explore the relationship between Definitions \ref{def:eij} and \ref{def:potentialcancel}.

\begin{lemma}
The number of monomials in $M_{i,j}$ which are \em not \em good is equal to $|E_{i, j}|$, except for $M_{0, 0}$, where the number of monomials which are not good is $|E_{0, 0}| - 1$.
\end{lemma}

\begin{proof}
Recall that a monomial $X^aY^b$ with $a + b \leq 2(q - 1)$ is not good when $q - 1 - a \leq_2 b$. 
From the definition of $\leq_2$, this means that $B(q - 1 - a) \subseteq B(b)$, or equivalently that $\overline{B(a)} \subseteq B(b)$, using the fact that $B(q - 1 - a)$ is the complement of $B(a)$.  Thus, for every place in which the binary expansion of $a$ has a $0$, the binary expansion of $b$ must have a $1$.

We can then think of dividing $B(b)$ into two subsets: $B(b) \setminus B(a)$ and $B(b) \cap B(a)$. 
Since $\overline{B(a)}\subseteq B(b)$, we have $\overline{B(a)} \subseteq B(b) \setminus B(a)$.  Moreover, $\overline{B(a)} = [\ell] \setminus B(a) \supseteq B(b) \setminus B(a)$, and so $B(q -1 -a) = \overline{B(a)} = B(b) \setminus B(a)$.  That is, 
the value whose binary expansion is given by $B(b) \setminus B(a)$ is precisely $q - 1 - a$.
 Let $c$ be the value so that $B(c) = B(b) \cap B(a)$. 
 Then $B(b)$ is the disjoint union of $B(c)$ and $B(b) \setminus B(a) = B(q - 1 - a)$.
 In other words, $b = q-1-a + c$, and so 
 \[b \equiv c - a \mod{q - 1}.\] 
 Because $b + a \equiv j \mod{t}$ and $t \mid (q - 1)$, this implies that $c \equiv j \mod{t}$. 
 Moreover, $c \leq_2 a$. In fact, any choice of $c \leq_2 a$ and $c \equiv j \mod{t}$ would yield a unique $b$ so that $X^aY^b \in M_{i, j}$ and $X^aY^b$ is not good. But $E_{i, j}$ exactly counts how many pairs of $a, c$ there are which meet these requirements!

This implies that the number of monomials in $M_{i,j}$ which are not good (or else which are our special case of $P_*(X,Y) = X^{q-1}Y^{q-1}$ from Remark~\ref{rem:special}) is equal to $|E_{i,j}|$. To deal with this special case, recall that $X^{q - 1}Y^{q - 1}$ is good even though it does not fit nicely into our characterization; thus the logic above counts it as ``not good." Because $s$ and $t$ both divide $q - 1$, this monomial is in $M_{0, 0}$. Thus, the number of monomials in $M_{0, 0}$ which are not good is $|E_{0, 0}| - 1$. \end{proof}

Armed with our new notation and the above lemma, we will proceed to state and prove a general bound on the dimension of $\mathcal{F}_{s, t}$.

\begin{theorem}[Theorem~\ref{thm:equiv} restated]\label{thm:gendim}
	Let $\mathcal{F}_{s,t}$ be the partial lift of $\mathcal{F}_0$ with respect to $\mathcal{L}_{s,t}$, as in Construction~\ref{cons:theconstruction}; that is, $\mathcal{F}_{s,t}$ is the set of polynomials in $\F_q[X,Y]$ that restrict nicely to all lines $L$ in $\Lst$.  Then the dimension $\mathcal{F}_{s, t}$ is at least $q^2 - e(s, t)$, where $e(s,t)$ is as in Definition~\ref{def:eij}.
\end{theorem}

\begin{proof}
	We exhibit a large linearly independent set of polynomials (both monomials and binomials) contained in $\mathcal{F}_{s,t}$.  
Let us begin building up our linearly independent set by adding in all the good monomials. These restrict nicely on every line, so they certainly restrict nicely on $\Lst$. Moreover, they are all linearly independent. For convenience, let us call the number of good monomials $g$. (The actual value was calculated in Corollary \ref{cor:goodcount}, but will not matter here.)

Now, we would like to count how much more we get by adding in binomials of the form \eqref{eq:nicebin2}. We can find a linearly independent set of binomials by considering each $M_{i, j}$ separately. Within each equivalence class, we can select a single representative monomial, and combine it with every other monomial in that class to create our set of binomials. However, we only want to combine two monomials when neither are already good.
A binomial formed by adding two good monomials \em will \em restrict nicely, but is not linearly independent from the set of good monomials which we have already included. A binomial formed by adding one good monomial with one which is not good will not actually restrict nicely at all. 
Thus, for every pair $(i, j)$, we will look at the subset of $M_{i, j}$ which is not good, pick a representative monomial, and combine it with every other monomial. In fact, we know that there are $|E_{i, j}|$ not-good monomials in $M_{i, j}$ (putting aside $i = j = 0$ for now). Thus whenever $|E_{i, j}| \neq 0$ we can form $|E_{i, j}| - 1$ binomials which restrict nicely, and which are linearly independent both with each other and with the good monomials. The only exception here is that when $i = j = 0$, we must check if $|E_{0, 0}| > 1$, and if so we can form $|E_{0, 0}| - 2$ binomials.

Summing over all these equivalence classes, the above reasoning shows that there is a set $A$ of binomials of the form \eqref{eq:nicebin2}
of size at  least 
$$\sum_{|E_{i, j}| \neq 0}( |E_{i, j}| - 1) - 1 = \sum_{i, j} |E_{i, j}| - 1 - \sum_{|E_{i, j}| \neq 0} 1,$$
so that the elements of $A$ restrict nicely to all lines in $\Lst$, are linearly independent, and moreover are linearly independent from the $g$ good monomials already accounted for.

Above, we have subtracted 1 to account for the case where $i = j = 0$. (If that set is completely empty, this correction is not necessary, but as we are only seeking a lower bound we can simply subtract 1 anyhow.) But we note that $(\sum_{i, j}|E_{i, j}|) - 1$ is exactly the number of not-good monomials across all equivalence classes! In particular, the number of not good monomials is $q^2 - g$, since there are $q^2$ monomials total. Moreover, the count of non-empty classes $E_{i, j}$ is exactly $e(s, t)$. Then we can see that, in fact, our expression for the number of binomials simplifies to $q^2 - g - e(s, t)$.

Now consider the set $S \subseteq \F_q[X,Y]$ consisting of the polynomials in $A$ along with the good monomials.
The reasoning above shows that
\[ |S| = g + |A| \geq g + q^2 - g - e(s,t) = q^2 - e(s,t). \]
Further, since the elements of $S$ are all linearly independent and $S \subseteq \mathcal{F}_{s,t}$, we have established that $\dim(\mathcal{F}_{s,t}) \geq q^2 - e(s,t)$, as desired.
\end{proof}

We now have an expression lower bounding the dimension of our code, but our expression depends on $e(s, t)$. We would like to know that $e(s, t)$ is not too big. It is easy to see that $e(s, t) \leq st$, because there are only $st$ choices for $(i, j)$. Moreover, we know that $e(s, t) \leq q^2 - g = 3^\l - 1$, the total number of not-good monomials. 
As we will see in Section~\ref{sec:constructions}, this first bound $e(s,t) \leq st$ is nontrivial, and can in fact recover the result of $N - K = s \sqrt{N}$ of \cite{codedpir}.  However, the point of all this work is that in fact we will be able to choose
 $s$ and $t$ so that we can get a much tighter bound on $e(s, t)$, as we will explore in the next section.  Using Observation~\ref{obs:keyobs}, this will result in constructions of high-rate codes with the $s$-DRGP, and will prove Theorem~\ref{thm:main}.

\section{Instantiations}\label{sec:constructions}
In this section, we examine a few specific choices of $s$ and $t$.  
We will focus (in Section~\ref{sec:easyt}) on the case where $t = q-1$ and $s$ is a proper divisor of $q-1$.  This will immediately allow us to obtain codes with the $s$-DRGP and with codimension $N - K \leq s \sqrt{N}$, recovering the result of \cite{codedpir}.  However, in this setting we will also be able to get a much tighter bound on $e(s,t)$,: we will specialize to a particular choice of $s$ in Section~\ref{sec:mainconstruction}, and this will allow us to establish the codes claimed in Theorem~\ref{thm:main}.
Finally, in Section~\ref{sec:variants}, we discuss other possible instantiations of our framework of Section~\ref{sec:framework}.  While our main quantitative results come from the first setting of Section~\ref{sec:mainconstruction}, the point of Section~\ref{sec:variants} is more to make the conceptual point that our approach applies more generally.  We hope that some of these techniques might find other applications.

\subsection{The case where $t = q - 1$}\label{sec:easyt}

One of the simplest choices we can make within our framework is to set $t = q - 1$, while $s | q-1$ is any divisor.  That is, we consider all simple lines $L(T) = (T, \alpha T + \beta)$ where $\beta$ may vary over all of $\F_q^*$, and where $\alpha \in G_s$ lives in a multiplicative subgroup of $\F_q^*$.
One reason that this choice is convenient is that it is easy to understand the number of disjoint repair groups.

\begin{lemma}\label{lem:repgrps}
Let $(x, y) \in \F_q^2 \setminus \inset{(0,0)}$. Then there exist at least $s - 1$ lines in $\L_{s, q - 1}$ which pass through $(x, y)$. Further, no lines in $\L_{s, q - 1}$ pass through $(0, 0)$.
\end{lemma}

\begin{proof}
	Let $(x,y) \in \F_q^2$ be nonzero.  Then for each $\alpha \in G_s$, there is some $\beta = y - \alpha x$ so that $L_{\alpha,\beta}(x) = (x,y)$, where $L_{\alpha,\beta}(T) = (T, \alpha T + \beta)$.  All of these parameterizations $L_{\alpha,\beta}$ are non-equivalent, and at most one out of $s$ lines has $\beta = 0$.  Thus, the remaining $s - 1$ such lines have $L_{\alpha, \beta}(x) = (x,y)$, and $L_{\alpha,\beta} \in \L_{s,q-1}$.
	Finally, because any simple line $L_{\alpha,\beta}$ passing through the origin must have $\beta  = 0$, there are no such lines in $\L_{s,q-1}$.
\end{proof}

Observation~\ref{obs:keyobs} immediately implies that the code $\mathcal{F}_{s,q-1}$ has the $(s-1)$-DRGP.  It remains only to bound its dimension.
Theorem~\ref{thm:gendim}, along with the observation of the previous section that $e(s,q-1) \leq s(q-1)$ trivially, immediately implies DRGP codes that match the results of \cite{codedpir}:
\begin{corollary}\label{thm:oldresult}
	Let $q = 2^\ell$ and let $s | q-1$.  Then there is a linear code $\cC$ over $\F_q$ of length $N = q^2 - 1$ with dimension
	\[ K \geq N + 1 - s \inparen{\sqrt{N+1} - 1} = N - O(s \sqrt{N})\]
	so that $\cC$ has the $(s - 1)$-DRGP.
\end{corollary}
\begin{proof}
	We consider the code 
	\[ \cC= \inset{\langle P(x,y)\rangle_{(x,y)\in \F_q^2 \setminus \inset{(0,0)}}  \suchthat P \in \mathcal{F}_{s,q-1} }.\]
	The fact that $\cC$ has the $(s-1)$-DRGP follows from the same reasoning as Observation~\ref{obs:keyobs}, combined with the fact from Lemma~\ref{lem:repgrps} that every point in $\F_q^2 \setminus \inset{(0,0)}$ lies on at least $s-1$ lines in $\L_{s,q-1}$, and that these lines are not equivalent.
	The only point of concern is the fact that the origin is not included in the evaluation points; what if one of the lines (which is needed for the repair groups) passes through the origin?  However, Lemma~\ref{lem:repgrps} assures us that none of the lines in $\mathcal{L}_{s,q-1}$ pass through the origin, so the argument about repair groups still holds.  Finally, the claim about the dimension follows from Theorem~\ref{thm:gendim} using $N = q^2 - 1$ and $e(s,q-1) \leq s(q-1)$.  
	\end{proof}

Thus, by choosing $t = q-1$, we can recover the quantitative results of~\cite{codedpir}.  However, as we will see in the next section, by choosing $s$ carefully we can actually get a tighter bound on $e(s,t)$, and this will allow us to go beyond Corollary~\ref{thm:oldresult}.
\subsection{Choosing $s$: $s = \sqrt{q} - 1$, $t = q - 1$}\label{sec:mainconstruction}
In this section, we continue to make the choice $t = q-1$, and we will choose $s$ carefully in order to prove our main quantitative result of Theorem~\ref{thm:main}.  For the reader's convenience, we duplicate the theorem below.
\begin{theorem}[Theorem~\ref{thm:main} restated]
	Suppose that $q = 2^\ell$ for even $\ell$, and let $N = q^2 - 1$.  There is a linear code $\cC$ over $\F_q$ of length $N$ and dimension 
	\[ K \geq N - O(N^{.714}) \]
	which has the $s$-DRGP for $s = \sqrt{q} - 2 =  (N + 1)^{1/4} - 1$.
\end{theorem}
Above in the proof of Corollary~\ref{thm:oldresult}, we were able to give a general bound for any $s | q-1$.
However, our analysis, which used the bound $e(s,t) \leq st$, was quite loose.  
In this section, we analyze the particular case when $s = \sqrt{q} - 1$, which will allow us the achieve the better bound of Theorem~\ref{thm:main}.
Given the proof of Corollary~\ref{thm:oldresult}, Theorem~\ref{thm:main} follows immediately from the following bound on $e(\sqrt{q} - 1, q - 1)$.
\begin{theorem}\label{thm:bigresult}
	Let $q = 2^{\ell}$ be an even power of $2$.  Then
\[e(\sqrt{q} - 1, q - 1) = O\inparen{(5 + \sqrt{5})^{\ell/2}}.\] 
\end{theorem}
Before we prove Theorem~\ref{thm:bigresult}, we briefly comment on why it suffices to prove Theorem~\ref{thm:main}.  As in the proof of Corollary~\ref{thm:oldresult}, we choose $\cC$ to be the code corresponding the $\mathcal{F}_{\sqrt{q}-1,q-1}$, with the origin punctured.  As before, $\cC$ immediately has the $(\sqrt{q} - 2)$-DRGP.  Moreover, the dimension of $\cC$ is at least
\begin{align*}
 K &\geq  q^2 - e(\sqrt{q} -1, q - 1) \\
 &\geq 2^{2\ell} - O\inparen{ (5 + \sqrt{5})^{\ell/2} } \\
  &\geq 2^{2\ell} - O\inparen{2^{ 2\ell \cdot \log_2(5 + \sqrt{5})/4 }} \\
  &\geq (N+1) - O\inparen{ (N+1)^{ \log_2(5 + \sqrt{5})/4} } \\
  &\geq N - O(N^{.714}),
  \end{align*}
where the last line follows from the computation $\log_2(5 + \sqrt{5})/4 \approx 0.7138$.
This establishes Theorem~\ref{thm:main}, modulo the proof of Theorem~\ref{thm:bigresult}.

\begin{proof}[Proof of Theorem \ref{thm:bigresult}]
Let $s = \sqrt{q} - 1$.
Let us recall what $e(s, q - 1)$ is counting: 
\[ e(\sqrt{q} - 1, q-1) = \inabs{\inset{(i,j) \in [\sqrt{q} - 1] \times [q-1] \suchthat \exists a \equiv_s i, c \equiv_{q-1} j, c \leq_2 a} }. \]
That is, we are counting the number of pairs $(i,j)$ so that $i < \sqrt{q}-1$, $j < q-1$, so that there exists some $a$ equivalent to $i$ mod $s$, and so that $c \leq_2 a$.
However, as $c \in \inset{0,\ldots, q-1}$, and $j < q-1$, the condition requires $c = j.$  Thus, somewhat more concisely we are counting
\begin{equation}
\label{eq:tocount}
e(\sqrt{q} - 1, q-1) = \inabs{\inset{(i,c) \in [\sqrt{q} - 1] \times [q-1] \suchthat \exists a \equiv_s i,  c \leq_2 a} }. 
\end{equation}

\begin{definition}
	\label{def:valid}
	For $i < \sqrt{q} - 1$, $c < q-1$, we say that the pair $(i,c)$ is a \em valid pair \em if there exists some $a < q$ so that $c \leq_2 a$ and so that $a \equiv_s i$.  We say that $a$ is a \em witness \em for $(i,c)$'s validity.   
\end{definition}
That is, $(i,c)$ is valid if it is counted in the right hand side of \eqref{eq:tocount}.
Thus, we wish to bound the number of valid pairs.  
However, as we will see, instead of counting the valid pairs $(i,c)$ directly, it will be easier to transform them into valid \em triples, \em $(y,c_0,c_1),$ and count those.  We define a valid triple below in Definition~\ref{def:triple}; the idea is to break up the length-$\ell$ binary expansion of $c$ into two parts, $c_0$ and $c_1$ of length $\ell/2$.  Additionally, it will be convenient to think of $i < 2^{\ell/2} - 1$ as a number $y \in [2^{\ell/2}]$ (that is, $y$ may take on the value $2^{\ell/2} - 1$).
The argument will proceed by induction on the length of the binary expansion of these triples, so in the definition below we will let their length be an arbitrary integer $r$, rather than $\ell/2$.  
%
\begin{definition}\label{def:triple}
	Let $r > 0$ be an integer.  A triple $(y,c_0,c_1) \in [2^r]^3$ is a valid triple of length $r$ if there exists some $a_0,a_1 \in [2^r]$ so that the following holds:
	\begin{itemize}
		\item $a_0 \geq_2 c_0$ and $a_1 \geq_2 c_1$; and
		\item\label{item:myequiv} $a_0 + a_1 \equiv y \mod{ 2^r}.$
	\end{itemize}
	We say that such an $(a_0,a_1)$ is a witness for $(y,c_0,c_1)$'s validity.
\end{definition}
We illustrated Definition~\ref{def:triple} in Figure~\ref{fig:triple}.
\begin{figure}
	\begin{center}
	\begin{tikzpicture}[scale=.6]
	\draw (0,0) rectangle (8,1);
	\node at (4,.5) {$a_1 \geq_2 c_1$};
	\draw (0,1.5) rectangle (8,2.5);
	\node at (4,2) {$a_0 \geq_2 c_0$};
	\draw[thick] (-2,-.5) to (9,-.5);
	\draw (0,-2) rectangle (8,-1);
	\node at (4,-1.5) {$y$};
	\node at (-.7,.5) {$\mathbf{+}$};
	\draw (-1.2, -2) rectangle (-.2, -1);
	\node at (-.7, -1.5) {$d$};
	\draw [decorate,decoration={brace,amplitude=10pt},xshift=0pt,yshift=3pt] (0,2.5) -- (8,2.5)node [black,midway,yshift=18pt] {$r$};
	\end{tikzpicture}
	\end{center}
\caption{Visualization of a valid triple $(y,c_0,c_1)$ and a witness $(a_0,a_1)$.   Above, the boxes denote binary expansions of length $r$, with the highest order bit on the left.  $d \in \{0,1\}$ is a bit that represents a carry in the addition.  For $(y,c_0,c_1)$ to be valid, $d$ can be either $0$ or $1$, and can depend on the choice of witness $(a_0,a_1)$.}
\label{fig:triple}
\end{figure}
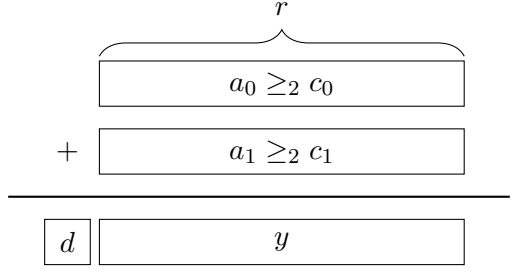
As indicated above, the reason we care about valid triples $(y,c_0,c_1)$ is that eventually for $r = \ell/2$ they will allow us to count valid pairs $(i,c)$.
To see why this is plausible, suppose that $(y,c_0, c_1)$ is a valid triple of length $r = \ell/2$ with witness $(a_0, a_1)$, and consider $c = c_0 + c_1 2^{\ell/2}$ and $a = a_0 + a_ 1 2^{\ell/2}$.  Then $c \leq_2 a$, because the binary expansions of $c$ and $a$ are simply the concatenation of the binary expansions of $c_0,c_1$ and $a_0,a_1$ respectively.  Moreover, we have 
$a \equiv a_0 + a_1 \mod 2^{\ell/2} - 1$,
while $y \equiv a_0 + a_1 \mod 2^{\ell/2}$.
This slight difference in the moduli means that we don't immediately have a tight connection between valid triples and valid pairs; but we will see in Claim~\ref{claim:pairs} below that in fact getting a handle on valid triples will be enough to count the valid pairs.

\newcommand{\carry}{\textsf{yesCarry}}
\newcommand{\nocarry}{\textsf{noCarry}}
\newcommand{\eithercarry}{\textsf{maybeCarry}}

We will count the valid triples $(y, c_0,c_1)$ using induction on $r$.  
To do this, we will divide up the triples into three classes, based on the witnesses $(a_0,a_1)$.
Define $\nocarry(r)$ to be the set of valid triples $(y,c_0,c_1)$ as above so that the following occurs:
For every witness $(a_0,a_1)$ to the validity of $(y,c_0,c_1)$, we have $a_0 + a_1 \leq 2^r - 1$.
To see why we have called this class $\nocarry$, observe that
this is the case when the addition $a_0 + a_1$ (as in the second requirement in Definition~\ref{def:triple}), does not have a carry to the $r$'th position when viewed as the addition of two $r$-bit binary numbers.  In Figure~\ref{fig:triple}, this is the case where $d = 0$ for all choices of $a_0,a_1$.

Analogously, we define $\carry(r)$ to be the set of valid triples $(y,c_0,c_1)$ so that for every witness $(a_0,a_1)$ to the validity of $(y,c_0,c_1)$, $a_0 + a_1 > 2^r - 1$.
Finally, we define $\eithercarry(r)$ to be the set of valid triples $(y,c_0,c_1)$ so that there exists witnesses $(a_0,a_1)$ for $(y,c_0,c_1)$ with $a_0 + a_1 > 2^r - 1$, and there also exist witnesses with $a_0 + a_1 \leq 2^r - 1$.

With these definitions in place, we can make the following claim.
\begin{claim}\label{claim:induct}
	The following identities hold:
	\begin{itemize}
		\item $|\nocarry(r+1)| = 3 \cdot |\nocarry(r)|$.
		\item $|\carry(r+1)| = 3 \cdot |\nocarry(r)| + 6 \cdot |\carry(r)| + 4 \cdot | \eithercarry(r)|$.
		\item $|\eithercarry(r+1)| = |\nocarry(r)| + |\carry(r)| + 4 \cdot |\eithercarry(r)|$.
	\end{itemize}
\end{claim}
Before we prove Claim~\ref{claim:induct}, we show why it is sufficient to prove Theorem~\ref{thm:bigresult}.  This will explain why we chose to study triples rather than pairs.
First, we observe that Claim~\ref{claim:induct} allows us to count $|\nocarry(\ell/2)|, |\carry(\ell/2)|,$ and $|\eithercarry(\ell/2)|$.  Indeed, when
 $r = 0$, we have trivially that $|\nocarry(0)| = 1$,  and the other two classes are empty and so by induction we have
\begin{equation}\label{eq:counting}
\begin{pmatrix}
|\nocarry(\ell/2)|\\|\carry(\ell/2)|\\|\eithercarry(\ell/2)|
\end{pmatrix}
 = 
\begin{pmatrix}
3 & 0 & 0 \\
3 & 6 & 4 \\
1 & 1 & 4
\end{pmatrix}^{\ell/2}
\cdot
\begin{pmatrix}
1 \\
0\\
0
\end{pmatrix}.
\end{equation}
In order to turn these counts into a bound on the number of valid pairs, we use the following claim.
\begin{claim}\label{claim:pairs}
 The total number of valid pairs $(i,c)$ is bounded by
	\[ |\nocarry(\ell/2)| + |\carry(\ell/2)| + 2|\eithercarry(\ell/2)|. \]
\end{claim}
\begin{proof}
Let $(y, c_0, c_1)$ be a valid triple of length $\ell/2$. We say that $(y, c_0, c_1)$ \em supports \em a pair $(i, c)$ if $c = c_0 \cdot 2^{\ell/2} + c_1$ and there exists a witness $(a_0, a_1)$ to the validity of $(y, c_0, c_1)$ such that $a_0 + a_1 \equiv i \mod{2^{\ell/2} - 1}$. We will first show that every valid pair is supported by at least one valid triple. Then, we will show that each triple in $\nocarry(\ell/2)$ and $\carry(\ell/2)$ supports at most one pair, and each triple in $\eithercarry(\ell/2)$ supports at most two pairs. Together, these statements imply that there can be at most $|\nocarry(\ell/2)| + |\carry(\ell/2)| + 2|\eithercarry(\ell/2)|$ valid pairs.

Let $(i, c)$ be a valid pair. Then there must exist some $a \equiv i \mod{2^{\ell/2} - 1}$ so that $c \leq_2 a$. Let $c_1 = c \mod{2^{\ell/2}}$ and let $c_0 = (c - c_1)/2^{\ell/2}$. Then $c = c_0 \cdot 2^{\ell/2} + c_1$. Similarly, let $a_1 = a \mod{2^{\ell/2}}$ and $a_0 = (a - a_1)/2^{\ell/2}$. Notice that the binary representations of $c_0$, $c_1$, $a_0$, and $a_1$ are all $\ell/2$ bits long, and moreover $a_0 \geq_2 c_0$ and $a_1 \geq_2 c_1$. We can form a valid triple $(y, c_0, c_1)$ where $y = a_0 + a_1 \mod{2^{\ell/2}}$, with $(a_0, a_1)$ as a witness. Moreover, we know that $a = a_0 \cdot 2^{\ell/2} + a_1 \equiv a_0 + a_1 \mod{2^{\ell/2} - 1}$. But $a \equiv i \mod{2^{\ell/2} - 1}$, so $a_0 + a_1 \equiv i \mod{2^{\ell/2} - 1}$. Thus we have constructed a valid triple which supports $(i, c)$.

Now, we would like to show that a triple $(y, c_0, c_1) \in \nocarry(\ell/2)$ can support at most one pair. Let $(i, c)$ and $(i', c')$ be two pairs supported by $(y, c_0, c_1)$; we will prove that these two pairs must in fact be equal. We immediately know that $c = c_0 \cdot 2^{\ell/2} + c_1 = c'$, but we still must show $i = i'$. Let $(a_0, a_1)$ be the witness to $(y, c_0, c_1)$ so that $a_0 + a_1 \equiv i \mod{2^{\ell/2} - 1}$, and let $(a_0', a_1')$ be the witness so that $a_0' + a_1' \equiv i' \mod{2^{\ell/2} - 1}$. Both such witnesses must exist because $(y, c_0, c_1)$ supports $(i, c)$ and $(i', c')$. Moreover, $a_0 + a_1 \equiv y \equiv a_0' + a_1' \mod{2^{\ell/2}}$, and because $(y, c_0, c_1) \in \nocarry(\ell/2)$, we know that $a_0 + a_1 \leq 2^r - 1$ and $a_0' + a_1' \leq 2^r - 1$. Then $a_0 + a_1 = a_0' + a_1'$, so $i = i'$, as desired.

Similarly, let $(y, c_0, c_1) \in \carry(\ell/2)$, and let $(i, c)$, $(i', c')$, $a_0, a_1$, and $a_0', a_1'$ as before. Once again, we know that $c = c'$ but must show that $i = i'$. As $a_0, a_1, a_0', a_1' < 2^{\ell/2}$ and $(y, c_0, c_1) \in \carry(\ell/2)$, so $2^{\ell/2} - 1 < a_0 + a_1 < 2^{\ell/2 + 1}$ and similarly for $a_0' + a_1'$. Then $a_0 + a_1 \equiv y \equiv a_0' + a_1' \mod{2^{\ell/2}}$ implies that $a_0 + a_1 = a_0' + a_1'$, so indeed $i = i'$.

Finally, we would like to show that a triple $(y, c_0, c_1) \in \eithercarry(\ell/2)$ can support at most two pairs. Now, given a witness $(a_0, a_1)$ to the validity of $(y, c_0, c_1)$, it is possible that $a_0 + a_1 > 2^{\ell/2} - 1$ and also possible that $a_0 + a_1 \leq 2^{\ell/2} - 1$. However, once we know which case a particular witness falls into, we know from above that supported pair associated with that witness is fully determined. Because there are only two cases, a triple in $\eithercarry(\ell/2)$ can support at most two pairs.

We have shown that at most $|\nocarry(\ell/2)| + |\carry(\ell/2)| + 2|\eithercarry(\ell/2)|$ valid pairs are supported by valid triples, and also that every valid pair must be supported by a valid triple. Thus there can be at most $|\nocarry(\ell/2)| + |\carry(\ell/2)| + 2|\eithercarry(\ell/2)|$ valid pairs.
\end{proof}

Finally, Claim~\ref{claim:pairs}, along with \eqref{eq:counting}, implies that the total number of valid pairs is bounded by
\begin{equation*}
\begin{pmatrix}
1 & 1 & 2
\end{pmatrix}
\cdot
\begin{pmatrix}
3 & 0 & 0 \\
3 & 6 & 4 \\
1 & 1 & 4
\end{pmatrix}^{\ell/2}
\cdot
\begin{pmatrix}
1 \\
0\\
0
\end{pmatrix} =: \begin{pmatrix}
1 & 1 & 2
\end{pmatrix} \cdot M^{\ell/2} \cdot \begin{pmatrix}
1 \\
0\\
0
\end{pmatrix}.
\end{equation*}
The matrix $M$ is diagonalizable and so we can compute $P$ and $\Lambda$ so that $M = P\Lambda P^{-1}$, which yields $M^{\ell/2} = P \Lambda^{\ell/2} P^{-1}$.
We compute
\begin{equation*}
M^{\ell/2} = 
\begin{pmatrix}
-1 & 0 & 0 \\
1 & 1 + \sqrt{5} & 1 - \sqrt{5} \\
0 & 1 & 1
\end{pmatrix}
\cdot
\begin{pmatrix}
3^{\ell/2} & 0 & 0 \\
0 & (5 + \sqrt{5})^{\ell/2} & 0 \\
0 & 0 & (5 - \sqrt{5})^{\ell/2}
\end{pmatrix}
\cdot
\begin{pmatrix}
-1 & 0 & 0 \\
\frac{1}{2\sqrt{5}} & \frac{1}{2\sqrt{5}} & \frac{5 - \sqrt{5}}{10} \\
-\frac{1}{2\sqrt{5}} & -\frac{1}{2\sqrt{5}} & \frac{5 + \sqrt{5}}{10}
\end{pmatrix}
\end{equation*}
This immediately bounds the number of valid pairs by $O\inparen{( 5 + \sqrt{5})^{\ell/2}}$, since that is the largest eigenvalue of $M$ in the decomposition above.  By the discussion about, this shows that
\[ e( \sqrt{q} - 1, q - 1) = O\inparen{( 5 + \sqrt{5})^{\ell/2}},\]
as desired.  This establishes Theorem~\ref{thm:bigresult}, except for the proof of Claim~\ref{claim:induct}, which we prove below.
\begin{proof}[Proof of Claim~\ref{claim:induct}]
Consider a valid triple $(y',c'_0, c'_1)$ of length $r + 1$, with witness $a_0', a_1'$.
	Let $(y,c_0,c_1)$ be the triple of length $r$ that is obtained by dropping the first (most significant) bit from each of $y,c_0,c_1$, and consider $a_0,a_1$ formed from $a_0', a_1'$ in the same way.  That is, to obtain $x \in [2^r]$ from $x' \in [2^{r+1}]$, we let $x = x' \mod{2^r}$.
	
	We observe that $(y,c_0,c_1)$ is in fact a valid triple of length $r$, and that $a_0, a_1$ is a witness for this validity.
	Indeed, first notice that since we are just dropping bits, the conditions about $2$-shadows are preserved.  Next, we have
	\begin{align*}
	a_0' + a_1' &\equiv y' \mod{2^{r+1}} \\
	a_0' + a_1' &\equiv y \mod{2^r} \qquad\qquad \text{since $2^r | 2^{r+1}$} \\
	a_0 + a_1 &\equiv y \mod{2^r} \qquad\qquad \text{using the definition of $a_0,a_1,y.$}	
		\end{align*}
	
	With this connection in mind, we go the other way.  Given a valid triple $(y,c_0,c_1)$ of length $r$, what are the valid triples $(y', c_0', c_1')$ of length $r+1$ that extend it?\footnote{Here we say that $x' \in [2^{r+1}]$ \em extends \em $x \in [2^r]$ if $x' = b 2^r + x$ for some $b \in \{0,1\}$; that is, the binary expansion of $x'$ is an extension of the binary expansion of $x$.}  The reasoning above establishes that any witness $a_0', a_1'$ for $(y', c_0', c_1')$ must extend a witness $a_0,a_1$ of $(y,c_0,c_1)$.

	We wish to understand the number and nature of the extensions of $(y,c_0,c_1)$.
	Suppose that $c_0' = b_0 2^r + c_0$ and $c_1' = b_1 2^r + c_1$, for $b_0,b_1 \in \{0,1\}$.   We illustrate the set-up in Figure~\ref{fig:possibilities}.
	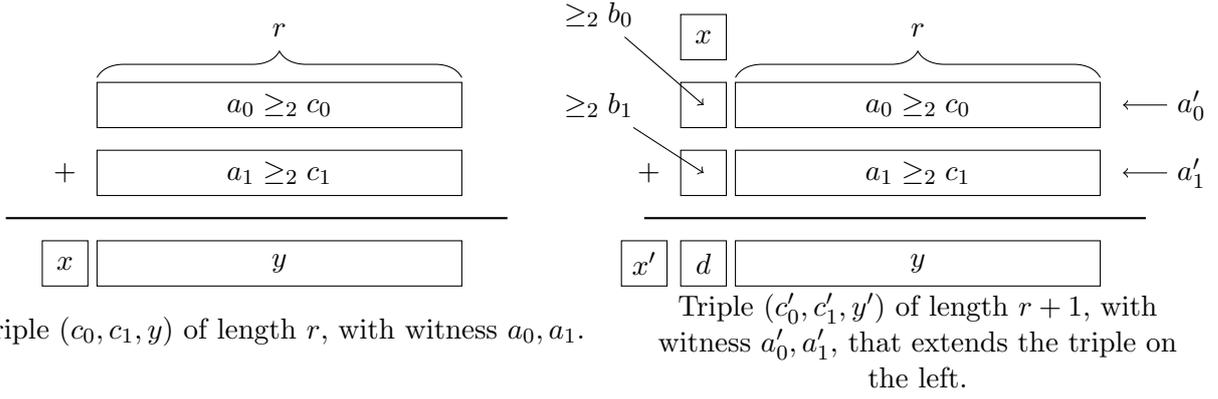
\begin{figure}
		\begin{center}
			\begin{tikzpicture}[scale=.6]
			\draw (0,0) rectangle (8,1);
			\node at (4,.5) {$a_1 \geq_2 c_1$};
			\draw (0,1.5) rectangle (8,2.5);
			\node at (4,2) {$a_0 \geq_2 c_0$};
			\draw[thick] (-2,-.5) to (9,-.5);
			\draw (0,-2) rectangle (8,-1);
			\node at (4,-1.5) {$y$};
			\node at (-.7,.5) {$\mathbf{+}$};
			\draw (-1.2, -2) rectangle (-.2, -1);
			\node at (-.7, -1.5) {$x$};
			\draw [decorate,decoration={brace,amplitude=10pt},xshift=0pt,yshift=3pt] (0,2.5) -- (8,2.5)node [black,midway,yshift=18pt] {$r$};
			\node at (4,-3) {Triple $(c_0,c_1,y)$ of length $r$, with witness $a_0,a_1$.};
			\begin{scope}[xshift=14cm]
	\draw (0,0) rectangle (8,1);
\node at (4,.5) {$a_1 \geq_2 c_1$};
\draw (0,1.5) rectangle (8,2.5);
\node at (4,2) {$a_0 \geq_2 c_0$};
\draw[thick] (-2,-.5) to (9,-.5);
\draw (0,-2) rectangle (8,-1);
\node at (4,-1.5) {$y$};
\node at (-1.9,.5) {$\mathbf{+}$};
\draw (-1.2, -2) rectangle (-.2, -1);
\node at (-.7, -1.5) {$d$};
\draw (-2.5, -2) rectangle (-1.5, -1);
\node at (-2, -1.5) {$x'$};
\draw (-1.2, 0) rectangle (-.2, 1);
\coordinate(b1) at (-.7, .5);
\node[black](l1) at (-3,2) {$\geq_2 b_1$};
\draw[black,->] (l1) to (b1);
\draw (-1.2, 1.5) rectangle (-.2, 2.5);
\coordinate(b2) at (-.7, 2);
\node[black](l2) at (-3, 4) {$\geq_2 b_0$};
\draw[black,->](l2) to (b2);
\draw (-1.2, 3) rectangle (-.2, 4);
\node at (-.7, 3.5) {$x$};
\draw [decorate,decoration={brace,amplitude=10pt},xshift=0pt,yshift=3pt] (0,2.5) -- (8,2.5)node [black,midway,yshift=18pt] {$r$};
\node at (4,-3.2) {\begin{minipage}{7cm}\begin{center}Triple $(c_0',c_1',y')$ of length $r+1$, with witness $a_0',a_1'$, that extends the triple on the left.\end{center}\end{minipage}};
\node(a) at (10,2) {$a_0'$};
\node(b) at (10,.5) {$a_1'$};
\draw[->] (a) to (8.5,2);
\draw[->] (b) to (8.5,.5);
			\end{scope}
			\end{tikzpicture}
		\end{center}
		\caption{Visualization of a valid triple $(y,c_0,c_1)$ and a witness $a_0,a_1$, and its extension to a triple $(y',c_0',c_1')$, with the notation used in the proof of Claim~\ref{claim:induct}.   Above, the boxes denote binary expansions with the highest order bit on the left.  $x \in \{0,1\}$ is a bit that represents a carry in the addition of $a_0 + a_1$. $x' \in \{0,1\}$ is a bit that represents a carry in the addition $a_0' + a_1'$. The notation $\geq_2 b_0$ indicates that whatever bit goes here must be larger than $b_0$.  That is, if $b_0$ is $1$, then we must choose the high bit of $a_0'$ to be $1$, so that we respect the requirement that $a_0' \geq_2 c_0'$.}
		\label{fig:possibilities}
	\end{figure}
	
	How $(y', c_0', c_1')$ behaves depends on which of the classes $(y,c_0,c_1)$ was in to begin with.  
	To begin to analyze this, suppose that $(y,c_0,c_1) \in \nocarry(r)$.  Thus, the only witnesses $a_0,a_1$ for $(y,c_0,c_1)$ must have $a_0 + a_1 < 2^r$.  In Figure~\ref{fig:possibilities}, the carry bit $x$ is equal to $0$.
	Now we ask: how many ways are there to extend the triple $(y,c_0,c_1)$, and what classes do these extensions live in?
	There are eight choices for this extension, corresponding to the choices of $b_0,b_1, d$.
	For each of these $8$ choices, we ask whether or not the resulting extension is a valid triple of length $r+1$, and if so, does it live in $\nocarry(r+1), \carry(r+1),$ or $\eithercarry(r+1)$?
	
	In Figure~\ref{fig:possibilities}, this corresponds to asking whether or not we can find ways to fill in the boxes labeled $\geq_2 b_0$ and $\geq_2 b_1$ so that all the constraints in the figure are met.  If the only way we can fill these in results in $x' = 0$, then this generates an element of $\nocarry(r+1)$.  If the only way to fill these in results in $x' = 1$, then this generates an element of $\carry(r+1)$.  If there are ways to result in both $x'=0$ and $x'=1$, then this generates an element of $\eithercarry(r+1)$.  And if there is no way to fill in the boxes, then no new valid triple is generated.
	
	We count these below.  With the machinery above, this is a straightforward but tedious exercise.  We work out tbe logic for the case that $(y,c_0,c_1) \in \nocarry(r)$ as examples.  The results of all of these calculations are tabulated in Table~\ref{table:count}.
	
	\renewcommand{\arraystretch}{1.5}
	\begin{table}
		\begin{center}
		\begin{tabular}{|r|c|c|c|c|}
			\hline
		$(b_0,b_1)=$	& $(0,0)$ & $(0,1)$ & $(1,0)$ & $(1,1)$ \\
				\hline\hline
$(y, c_0,c_1) \in \nocarry(r)$&
 $\eithercarry(r+1)$ & $\carry(r+1)$ & $\carry(r+1)$ & $\carry(r+1)$ \\
		\hline
$(y, c_0,c_1) \in \carry(r)$&
$\carry(r+1)$ & $\carry(r+1)$ & $\carry(r+1)$ & not possible \\
		\hline
$(y, c_0,c_1) \in \eithercarry(r)$&
$\eithercarry(r+1)$ & $\carry(r+1)$ & $\carry(r+1)$ & $\carry(r+1)$ \\
		\hline
			\end{tabular}
		\vspace{.3cm}
		
		$d=0$
		\vspace{.5cm} 
		
				\begin{tabular}{|r|c|c|c|c|}
					\hline
			$(b_0,b_1)=$	& $(0,0)$ & $(0,1)$ & $(1,0)$ & $(1,1)$\\
			\hline\hline
			$(y, c_0,c_1) \in \nocarry(r)$&
			$\nocarry(r+1)$ & $\nocarry(r+1)$ & $\nocarry(r+1)$ & not possible \\
			\hline
			$(y, c_0,c_1) \in \carry(r)$&
			$\eithercarry(r+1)$ & $\carry(r+1)$ & $\carry(r+1)$ & $\carry(r+1)$ \\
			\hline
			$(y, c_0,c_1) \in \eithercarry(r)$&
			$\eithercarry(r+1)$ & $\eithercarry(r+1)$ & $\eithercarry(r+1)$ & $\carry(r+1)$ \\
			\hline
		\end{tabular}
			\vspace{.3cm}
	
	$d = 1$
	\end{center}
		\caption{The possible outcomes for different combinations of $(y,c_0,c_1)$, $b_0,b_1$ and $d$, as described in the proof of Claim~\ref{claim:induct}.}
		\label{table:count}
	\end{table}
	
	Suppose that $(y,c_0,c_1) \in \nocarry(r)$, and we try to extend it to $(y', c_0', c_1')$. From the definition of $\nocarry(r)$, the only witnesses $a_0,a_1$ for $(y,c_0,c_1)$ must have $a_0 + a_1 < 2^r$.  
	We consider each of the eight possibilities in turn below.  
	First, suppose that $d=0$; we consider the four possibilities for $b_0,b_1$.
		\begin{description}
			\item[ $b_0 = b_1 = 0$.]  In this case, there are two ways to choose extensions to the witness $a_0,a_1$ that will result in a witness $a_0',a_1'$ for $(y',c_0',c_1')$.  We may either choose $a_0' = 2^r + a_0$ and $a_1' = 2^r + a_1$, or we may choose $a_0' = a_0$ and $a_1' = a_1$.  Thus, there is a witness for $(y',c_0', c_1')$ so that $a_0' + a_1' \geq 2^{r+1}$ and there is a witness so that $a_0' + a_1' < 2^{r+1}$.  Thus, for this choice of $d,b_0,b_1$, we have $(y', c_0', c_1') \in \eithercarry(r+1)$.
			\item[$b_0 = 1, b_1 = 0$.]  In this case, we must choose $a_0' = 2^r + a_0$ in order to satisfy $a_0' \geq_2 c_0'$.  We also must choose $a_1' = 2^r + a_1$ so that $a_0' + a_1' = (2^r + a_0) + (2^r + a_1) = 2^{r+1} + (a_0 + a_1)$ has a zero in the $r$'th bit (to match $i'$).  Here, we are using that $a_1 + a_0 < 2^r$.  Thus, there is only one choice for a witness $a_0',a_1'$, and we have $a_0' + a_1' \geq 2^{r+1}$.  Thus, for this choice of $d,b_0,b_1$, we have $(y', c_0', c_1') \in \carry(r+1)$.
			\item[$b_1 = 0, b_0 = 1$.] This is similar to the previous case, and results in a different $(y', c_0', c_1') \in \carry(r+1)$.
			\item[$b_0 = b_1 = 1$.] The requirement that $a_0' \geq_2 c_0'$ and $a_1' \geq_2 c_1'$ implies that we must choose $a_0' = 2^r + a_0$ and $a_1' = 2^r + a_1$.   Thus, this results in yet another $(y', c_0', c_1') \in \carry(r+1)$. 
		\end{description}
		Next, suppose that $d=1$; again consider four possibilities for $b_0,b_1$.
		\begin{description}
			\item [$b_0 = b_1 = 0$.]  In this case, there are again two ways to choose extensions to the witness $a_0,a_1$ that will result in a witness $a_0',a_1'$ for $(y',c_0',c_1')$.  We may either choose $a_0' = 2^r + a_0$ and $a_1' =  a_1$, or we may choose $a_0' = a_0$ and $a_1' = 2^r + a_1$.  In both cases, we have $a_0' + a_1' < 2^{r+1}$, and so $(y', c_0', c_1') \in \nocarry(r+1)$. 
			\item [$b_0 = 1, b_1 = 0$.]  In this case, we must choose $a_0' = 2^r + a_0$ in order to satisfy $a_0' \geq_2 c_0'$.  We also must choose $a_1' = a_1$ so that $a_0' + a_1' = (2^r + a_0) +  a_1 = 2^r + (a_0 + a_1)$ has a one in the $r$'th bit (to match $i'$).  As above, we are using that $a_1 + a_0 < 2^r$.  Thus, there is only one choice for a witness $a_0',a_1'$, and we have $a_0' + a_1' < 2^{r+1}$.  Thus, for this choice of $d,b_0,b_1$, we have $(y', c_0', c_1') \in \nocarry(r+1)$.
			\item [$b_1 = 0, b_0 = 1$.] This is similar to the previous case, and results in a different $(y', c_0', c_1') \in \nocarry(r+1)$.
			\item [$b_0 = b_1 = 1$.] The requirement that $a_0' \geq_2 c_0'$ and $a_1' \geq_2 c_1'$ implies that we must choose $a_0' = 2^r + a_0$ and $a_1' = 2^r + a_1$.   In this case, we have reached a contradiction, because $a_1' + a_0' = 2^{r+1} + (a_0 + a_1)$ has a zero in the $r$'th position, while $y' = d2^r + y$ was supposed to have a $1$.  Thus, this case does not contribute to any of the sets.
		\end{description}
	We may continue in this way for the case that $(y,c_0,c_1) \in \carry(r)$ or $\eithercarry(r)$.  We omit the details and report the results in Table~\ref{table:count}.

Finally, we add all the cases up.  We see that the only contributions to $\nocarry(r+1)$ come from $\nocarry(r)$, and there are three of them, so 
\[ |\nocarry(r + 1)| = 3 \cdot|\nocarry(r)|.\]
There are many ways to obtain contributions to $\carry(r+1)$: three from $\nocarry(r)$, six from $\carry(r)$, and four from $\eithercarry(r)$.  Thus,
\[ |\carry(r+1)| = 3 \cdot |\nocarry(r)| + 6 \cdot |\carry(r)| + 4 \cdot |\eithercarry(r)|. \]
Finally, for contributions to $\eithercarry(r+1)$, we have one from $\nocarry(r)$, one from $\carry(r)$, and four from $\eithercarry(r)$.  Thus,
\[ |\eithercarry(r+1)| = |\nocarry(r)| + |\carry(r)| + 4\cdot|\eithercarry(r)|.\]
This completes the proof of Claim~\ref{claim:induct}.
\end{proof}
The completion of the proof of Claim~\ref{claim:induct} establishes Theorem~\ref{thm:bigresult}.
\end{proof}

\subsection{Other possible settings of $s$ and $t$}\label{sec:variants}
Above, we have analyzed the case where $s = \sqrt{q}-1$ and $t = q-1$.  There are several other settings which appear to work just as well.  For example, we may take $s = \sqrt{q} + 1$ and $t = q-1$.  Empirically this setting seems to work well, and we suspect that the analysis would follow similarly to the analysis in the previous section.

Another option is to consider $s = q-1$ and $t = \sqrt{q}-1$.  That is, we consider lines of the form
\[ \mathcal{L} = \inset{ (T, \alpha T + \beta) \suchthat \alpha \in \F_q^*, \beta \in G_t }. \]
It is not hard to see that every point (other than those with $x=0$) has $t$ lines in $\mathcal{L}$ passing through it, and moreover no lines in $\mathcal{L}$ pass through $(0,y)$ for any $y$.  Thus, Observation~\ref{obs:keyobs}, along with Theorem~\ref{thm:gendim}, implies that the resulting partially lifted code (punctured along the $y$-axis) has the $t$-DRGP with rate at least $N - O(e(q-1,t))$.  
It seems as though an analysis similar to that of Section~\ref{sec:mainconstruction} applies here, and again gives codes with a similar result as in Theorem~\ref{thm:main}.

A third option, whose exploration we leave as an open line of work, is taking neither $s$ nor $t$ to be $q-1$.  Empirically, it seems like taking $s = t = (q-1)/3$ may be a good choice when $3 | q-1$.  However, while Theorem~\ref{thm:gendim} applies, getting a handle on $e(s,t)$ in this case is not straightforward.  Moreover, it is not immediately clear how many lines pass through every point.   We would need to puncture the code carefully to ensure that we maintain the number of disjoint repair groups to apply Observation~\ref{obs:keyobs}.

The framework of partially lifted codes is quite general, and we hope that future work will extend or improve upon our techniques to fully exploit this generality.

\section{Conclusion}
We have studied the $s$-DRGP for intermediate values of $s$.  As $s$ grows, the study of the $s$-DRGP interpolates between the study of LRCs and LCCs, and our hope is that by understanding intermediate $s$, we will improve our understanding on either end of this spectrum.  Using a new construction that we term a ``partially lifted code," we showed how to obtain codes of length $N$ with the $s$-DRGP for $s = \Theta(N^{1/4})$, that have dimension $K \geq N - N^{.714}$.  This is an improvement over previous results of $N - N^{3/4}$ in this parameter regime.  We stress that the main point of interest of this result is not the exponent $0.714$, which we do not believe is tight for Question~\ref{q:main}; rather, we think that our results are interesting because (a) they show that one can in fact beat $N - O()s\sqrt{N})$ for $s = N^{1/4} \ll \sqrt{N}$, and (b) they highlight the class of partially lifted codes, which we hope will be of independent interest.

\section*{Acknowledgements} We thank Alex Vardy and Eitan Yaakobi for helpful exchanges.

\bibliographystyle{alpha}
\bibliography{refs}

\end{document}